\documentclass[a4paper,twocolumn,11pt,accepted=2024-11-15]{quantumarticle}
\pdfoutput=1
\usepackage[utf8]{inputenc}
\usepackage[english]{babel}
\usepackage[T1]{fontenc}
\usepackage{amsmath}
\usepackage{amssymb}
\usepackage{amsthm}
\usepackage{hyperref}
\usepackage{balance}
\usepackage{tikz}
\usepackage{lipsum}
\usepackage[numbers,sort&compress]{natbib}

\newcommand{\norm}[1]{\mathinner{||{#1}||_1}}
\newcommand{\normsup}[1]{\mathinner{||{#1}||_{\infty}} }
\newcommand{\normeu}[1]{\mathinner{||{#1}||}}

\newtheorem{theorem}{Theorem}[section]
\newtheorem{lemma}{Lemma}[section]
\newtheorem{definition}{Definition}[section]

\begin{document}

\title{Time dependent Markovian master equation beyond the adiabatic limit}

\author{Giovanni Di Meglio}
\affiliation{Institut für Theoretische Physik and IQST,
Albert-Einstein-Allee 11, Universität Ulm, D-89081 Ulm, Germany}
\email{giovanni.di-meglio@uni-ulm.de}
\orcid{0000-0003-3066-9293}
\author{Martin B. Plenio}
\orcid{0000-0003-4238-8843}
\affiliation{Institut für Theoretische Physik and IQST,
Albert-Einstein-Allee 11, Universität Ulm, D-89081 Ulm, Germany}
\author{Susana F. Huelga}
\affiliation{Institut für Theoretische Physik and IQST,
Albert-Einstein-Allee 11, Universität Ulm, D-89081 Ulm, Germany}
\orcid{0000-0003-1277-8154}

\maketitle

\begin{abstract}
We derive a Markovian master equation that models the evolution of systems subject to driving and control fields. Our approach combines time rescaling and weak-coupling limits for the system-environment interaction with a secular approximation.
The derivation makes use of the adiabatic time-evolution operator in a manner that allows for the efficient description of strong driving, while recovering the well-known adiabatic master equation in the appropriate limit.
To illustrate the effectiveness of our approach, firstly we apply it  to the paradigmatic case of a two-level (qubit) system subject to a form of periodic driving that remains unsolvable using a Floquet representation and lastly we extend this scenario to the situation of two interacting qubits, the first driven while the second one directly in contact with the environment.
We demonstrate the reliability and broad scope of our approach by benchmarking the solutions of the derived reduced time evolution against numerically exact simulations using tensor networks. Our results provide rigorous conditions that must be satisfied by phenomenological master equations for driven systems that do not rely on first-principles derivations. 
\end{abstract}

\section{Introduction}

Any realistic treatment of a quantum system cannot ignore the presence of additional interacting systems that are outside of our control, commonly denoted as the system's environment. When the additional degrees of freedom of the environment are taken into account, the unitary evolution of a closed quantum system must be modified to account for the system-environment interactions, which may cause e.g. decoherence and dissipation.

To obtain a description of the dynamics of the open system, we can derive and solve master equations by tracing out the environmental degrees of freedom from an initial system-environment Hamiltonian. The Gorini-Kossakowski-Lindblad-Sudarshan (GKLS) equation \cite{Lindblad, GKS} is a widely-used and celebrated example of master equation. Its standard microscopic derivation relies on several assumptions, including the Born-Markov and secular approximations, the latter correspondent to the requirement of large spacing for the system Bohr frequencies \cite{ Petruccione,Rivas_1,Rivas_2}. However, this derivation was essentially tailored to time-independent Hamiltonians and it becomes considerably more challenging for time-dependent cases. 
The primary obstacle was the decomposition of the interaction picture Hamiltonian in a manner that transforms the master equation into the Lindblad form, which guarantees complete positivity of the dynamics at all times. As a result, less research was initially conducted in this area and only during the last decade prominent results started to appear in the literature.

One of the earliest attempts to address this problem was the pioneering work by Davies and Spohn in \cite{Davies_3}. In this work, the authors proposed a first extension of the time-independent framework assuming sufficiently weak driving, in order to formulate a linear response theory in the presence of a generic Markovian environment. Shortly after, Alicki \cite{Alicki_0} employed a Markovian master equation with an external Hamiltonian perturbation as paradigmatic model of a quantum heat engine, assuming that the driving is sufficiently slow to be considered constant on a time-scale much larger than the bath relaxation time but shorter than the characteristic system dynamics. Under these conditions, the conventional time-independent kernel can be replaced with an analogous expression written in the instantaneous eigenbasis of the system Hamiltonian. Subsequent analyses include results derived in connection with fault-tolerance in adiabatic quantum computation and quantum error correction, as in \cite{Childs, Alicki_1}. In the following years, apart from some phenomenological approaches, only two microscopic derivations have been studied systematically. The first one by Albash et al. \cite{Zanardi} requires the adiabatic condition on the closed system dynamics, therefore named adiabatic master equation after that, and has found a wide range of applications and further developments so far, especially for studying the impact of noise on adiabatic dynamics (as in \cite{Nalbach}); the second concerns the periodic driving treated by means of Floquet theory \cite{Chu_1,Hanggi_1, Breuer, Chu_2, Alicki_3,Floquet,Hartmann}, that allows to extend the ordinary microscopic derivation without any additional complication, with the drawback of being in practice applicable only to very few models whose Floquet representation is known.

Recently, starting from the work by Yamaguchi et al. in \cite{Yamaguchi}, different solutions have been proposed in order to tackle the problem for strongly driven systems, i.e. beyond the adiabatic regime. We note in particular Dann et al. \cite{Dann}, who made use of a Lie algebraic structure  generated by the time-evolution operator in order to obtain a suitable Markovian master equation for arbitrary driving, which includes as limiting cases both the adiabatic master equation and the Floquet theory based approach, but also
Wang et al. \cite{Wang}, where a different master equation is derived heuristically by using the Nakajima-Zwanzig approach in the adiabatic reference frame and applied to the dissipative Landau-Zener model. 

Finally, we also mention an alternative line of works, aimed at 
deriving Markovian master equations circumventing the secular approximation (\cite{Wacker,Lidar,Davidovic, Nathan_1, Nathan_2, Konig}). Most of these approaches are based on manipulations of the Bloch-Redfield master equation in order to suitably overcome the violation of positivity, while keeping the so called secular contributions.

Despite the extensive literature on the topic, we contend that additional  methods may provide a wider understanding of sufficient conditions that allow to achieve a Markovian evolution in the presence of control or driving fields and eventually fill the gap where the previous master equations could be extremely hard to handle.

In this work, we derive a Markovian master equation whose validity holds even if the system is strongly driven. In particular, such master equation represents a good approximation, in a sense that will be precisely clarified, of the exact reduced dynamics, under the conditions of 
weak system-environment coupling, fast decay of the environment correlation functions and sufficiently large spacing for the system Bohr frequencies, as explained more in details in the following.

Our derivation starts from the 
Nakajima-Zwanzig equation and draws inspiration from the approach outlined by Davies in the seminal works
\cite{Davies_1,Davies_2} (see also \cite{Rivas_1, Spohn}). The crucial idea is to apply a time-rescaling with respect to the dimensionless system-environment coupling and then to study the uniform convergence of the system density matrix in the weak-coupling limit. For 
time-dependent Hamiltonians, we show that this rescaling can be consistently applied provided that the system parameters are renormalised to reabsorb the coupling dependence in the Hamiltonian. 
By implementing our rescaling procedure and utilising the weak-coupling limit for the environmental interaction, we can apply the method of stationary phase to introduce a secular approximation. This approach leads to a Markovian master equation in Lindblad form, featuring time-dependent jump operators and positive dissipation rates.
Moreover, this procedure also delivers an estimation of the range of validity of the master equation in terms of the coupling strength of the system-environment interaction and the system parameters. 

We substantiate these claims with the application of our master equation to the analysis of a periodically driven two-level system.
It is well known that, for periodic Hamiltonians, the Floquet theory based approach leads to a Lindblad master equation with a derivation that matches that for a time-independent scenario, as long as the Floquet representation of the system time-evolution operator is known (\cite{Floquet}).
However, in many cases this representation is hard to obtain, which 
limits significantly the practical implementation of this approach. 

In contrast, we demonstrate that our framework is specifically tailored for the effective study of strong periodic drivings. This approach provides a valuable alternative method to analyse the behaviour of populations and coherences in the instantaneous eigenbasis of the system Hamiltonian. 
In the strong driving regime, we show that the driven dissipative dynamics described by our master equation is responsible for enhanced generation of coherence as compared with the purely unitary dynamics, the latter in agreement with results provided in \cite{Dann}.
Ultimately, we check that in the adiabatic limit our master equation correctly reproduces the equation in \cite{Zanardi}.

As last step, we extend the analysis to the situation of two interacting qubits, with one of them subject to a periodic driving and the other in direct contact with the environment. We analyse the combined effect of driving and dissipation in the strong driving regime, by looking at local observables of the two qubits. Furthermore, the limiting case of weakly interacting qubits is treated in detail. In particular, it is shown in this regime that the master equation fails to correctly reproduce the dynamics of the system, as result of the breakdown of the full-secular approximation. 

Our results are supported by numerically exact simulations obtained by propagating the joint dynamics of system and environment using tensor networks. 

This paper is organised as follows. We begin by presenting the general derivation of the quantum master equation (section 2). We then move to discuss the application to the paradigmatic case of a periodically driven qubit coupled to a bosonic environment and following the case of two interacting qubits, proceeding in addition to benchmark the results of our derived master equations against numerically exact simulations (section 3 and 4). We finally present the conclusions and give further prospects on the present topic (section 5).

A quick summary of the main results obtained are reported in the paragraph at the end of section 2, where the reader can find the necessary information to calculate and make use of the master equation derived in this paper.

Additional technical details on the considered methodology can be found in the appendices (A-C), including a succinct discussion on tensor network simulations (appendix D).

\section{Derivation of the reduced system's dynamics. General theory}

For completeness, let us start by giving some definitions concerning dynamical maps  that will be used throughout this paper, mainly in order to fix the notation. The reader familiar with the topic can directly jump to the subsection 'Characterisation of the problem' in the next page.
Let $\mathcal{H}$ be a separable Hilbert space and $L(\mathcal{H})=\{ X:\mathcal{H}\rightarrow \mathcal{H}| X \text{\space linear} \}$ the space of linear operators acting on the Hilbert space, which we can equip with the trace norm $\norm{X}=tr[\sqrt{X^{\dagger}X}]$, endowing the structure of Banach space $\mathcal{B}(\mathcal{H})=(L(\mathcal{H}),\norm{.})$ of trace-class operators. 
\begin{definition}
A quantum state (or density matrix) is a positive semidefinite ($\rho\ge 0$) and self-adjoint operator $\rho\in\mathcal{B}(\mathcal{H})$ with $\norm\rho=1$. The set of quantum states will be indicated by $\mathcal{B}^+(\mathcal{H})\subset \mathcal{B}(\mathcal{H})$.
\end{definition}
The evolution of a quantum system is determined by linear transformations belonging to $\mathcal{B}^*(\mathcal{H})=\{ T: \mathcal{B}(\mathcal{H})\rightarrow \mathcal{B}(\mathcal{H})| T \text{\space linear} \}$, which is a Banach space with respect to the norm $\normsup T=\sup\limits_{X\in \mathcal{B}(\mathcal{H}), X\neq 0}\frac{\norm{T(X)}}{\norm X}$, with the property $\normsup{T_1 T_2}\le \normsup{T_1}\normsup{T_2}$.
We state here some definitions of convergence for operators, that will be useful later.
\begin{definition}
Given a  Banach space $\mathcal{B}$ and a sequence of operators $\{T_n\}\in \mathcal{B}^*$, if some $T\in\mathcal{B}^*$ exists such that
\begin{equation}
\lim\limits_{n\rightarrow\infty}\normsup{T_n-T}=0,
\end{equation}
then we say that the sequence converges uniformly to $T$.
\end{definition}
\begin{definition}
Given a  Banach space $\mathcal{B}$ and a sequence of operators $\{T_n\}\in \mathcal{B}^*$, if some $T\in\mathcal{B}^*$ exists such that $\forall X\in\mathcal{B}$ we have
\begin{equation}
\lim\limits_{n\rightarrow\infty}\norm{T_n X-T X}=0,
\end{equation}
then we say that the sequence converges strongly to $T$.
\end{definition}
Note that strong convergence in $\mathcal{B}^*$ requires convergence in $\mathcal{B}$ with respect to the norm $\norm{.}$ for the sequence given by $\{ T_n X\} \in \mathcal{B}$. Moreover, uniform convergence
in $\mathcal{B}^*$ always implies strong convergence in $\mathcal{B}^*$.
\begin{definition}
Let $L$ be a vector space,  two norms $||.||_1$ and $||.||_2$ on $L$ are equivalent if there are constants $C_1, C_2\ge 0$ such that $C_1||X||_1\le ||X||_2\le C_2||X||_1$, $\forall X\in L$.
\end{definition}
In finite dimensional spaces all norms are equivalent, therefore proving the convergence with respect to a given norm automatically guarantees the convergence in all other norms.

The transformations that map states into states need to satisfy some properties. 
\begin{definition}
$\Lambda\in \mathcal{B}^*(\mathcal{H})$ is a positive map if $\Lambda(X)\ge 0$ for any $X\ge 0$.
\end{definition}
\begin{definition}
$\Lambda\in \mathcal{B}^*(\mathcal{H})$ is k-positive if $\Lambda\otimes id_k : \mathcal{B}(\mathcal{H})\otimes \mathcal{M}_k(\mathbb{C}) \rightarrow \mathcal{B}(\mathcal{H})\otimes \mathcal{M}_k(\mathbb{C})  $ is a positive map, where $id_k$ is the identity in the space of $k\times k$ complex matrices $\mathcal{M}_k(\mathbb{C})$.
\end{definition}
\begin{definition}
$\Lambda\in \mathcal{B}^*(\mathcal{H})$ is completely positive (CP) if it is k-positive $\forall k\in \mathbb{N}$.
\end{definition}
Physically, the CP condition can be argued to be fundamental as the requirement of positivity alone is not sufficient to guarantee that a quantum state is mapped into a positive semidefinite operator, when we include the possibility of  correlations with additional degrees of freedom.
\begin{definition}
\label{def_map}
$\Lambda\in \mathcal{B}^*(\mathcal{H})$ is a quantum evolution map (or dynamical map) if it is completely positive (CP), trace preserving (TP) and Hermiticity preserving, i.e. $\Lambda(\mathcal{B}^+(\mathcal{H}))=\mathcal{B}^+(\mathcal{H})$. 
\end{definition}
We are going to consider now families of quantum evolutions which depend on some real parameter. Moreover, henceforth we assume a finite dimensional Banach space $\mathcal{B(H)}$. 
\begin{definition}
An evolution family is a differentiable family of two-parameter maps given by 
\begin{equation}
\Big\{ \Lambda_{t,s}=\mathcal{T}\exp\Big\{ \int\limits_s^t du \mathcal{L}(u)\Big\} \Big\}_{t,s\ge 0}, 
\end{equation}
where $\mathcal{T}$ is a time-ordering, $\mathcal{L}(t)\in \mathcal{B^*(H)}$, which satisfy for all $t\ge r\ge s$
\begin{equation}
\begin{aligned}
\Lambda_{t,s}&=\Lambda_{t,r}\Lambda_{r,s},   \\
\Lambda_{s,s}&=\mathbb{I} .\\
\end{aligned}
\end{equation}
The first property is frequently called divisibility.
\end{definition}
The expression for the evolution family has to be understood as a Dyson series, which is always convergent for finite dimensional spaces, provided the generator $\mathcal{L}(t)$ to be bounded.

Among the classes of possible evolutions for a quantum state described by evolution families we are interested in the following:
\begin{definition}
An evolution family $\{ \Lambda_{t,s}=\mathcal{T}\exp\Big\{ \int\limits_s^t du \mathcal{L}(u)\Big\} \}_{t,s\ge 0}$ which fulfills definition \ref{def_map} $\forall t,s$ with $t\ge s\ge 0$ is called Markovian or equivalently CP-divisible. 
\end{definition}
The familiar case where the dynamical map $\Lambda_{t,s}$ is assumed to obey a semigroup composition law for arbitrary intermediate $r$ is a special instance where the resulting dynamics becomes time homogeneous.

A full characterisation of the generator for Markovian quantum processes is provided by the following important result (see \cite{Rivas_1} for a derivation that generalises \cite{Lindblad,GKS}):
\begin{theorem}
\label{th_1}
The evolution family $\{ \Lambda_{t,s}=\mathcal{T}\exp\Big\{ \int\limits_s^t du \mathcal{L}(u)\Big\} \}_{t,s\ge 0}$ is Markovian if and only if 
\begin{equation}
\begin{aligned}
\label{gkls_gen}
\mathcal{L}(t)\rho=&-\imath[H(t),\rho]+\sum\limits_{k\in I}\gamma_k(t)\Big( V_k(t)\rho V_k^{\dagger}(t) \\
&-\frac{1}{2}\{ V_k^{\dagger}(t)V_k(t),\rho\}\Big) ,
\end{aligned}
\end{equation}
where $H(t)\in L(\mathcal{H})$ is self-adjoint, $V_k(t)\in L(\mathcal{H})$ are called jump operators, $\gamma_k(t)\ge 0$ $\forall k, t$ are the so-called dissipation rates and $I$ is a set of indices. $\mathcal{L}(t)$ in \eqref{gkls_gen} is also called GKLS generator.
\end{theorem}
Note that for a time homogeneous dynamics, jump operators and dissipative rates become time independent and we recover the functional form of the dissipator familiarly referred to as "Lindblad form". This ends our compilation of preliminary notions. 
Now we are ready to introduce our problem.

\textbf{Characterisation of the problem}.
Non-unitary dynamics as the one described by Theorem \ref{th_1} arises as the result of the interaction between our system of interest and additional degrees of freedom which are not under our control. 
Let $\mathcal{B}(\mathcal{H}_{S})$ be the finite dimensional Banach space of a first quantum system, simply called system, and $\mathcal{B}(\mathcal{H}_B)$ of a second one called environment, which needs not to be finite dimensional. We suppose that the evolution of any quantum state $\rho\in\mathcal{B}(\mathcal{H}_S\otimes\mathcal{H}_B)$ is described by the unitary dynamics
\begin{equation}
\rho\rightarrow \rho(t)=\mathcal{T}\exp\Big\{\int\limits_0^t ds Z(s)\Big\}\rho ,
\end{equation}
with generator $Z(t)\rho=-\imath[H(t),\rho]$  and Hamiltonian
\begin{equation}
\label{general_h}
H(t)=H_S(t)+H_B+gH_I,
\end{equation}
where $g$ is a dimensionless coupling constant that defines the strength of
the system-environment interaction, $H_S(t)\in L(\mathcal{H}_S)$ analytic in $t$ and non-degenerate, $H_B\in L(\mathcal{H}_B), H_I\in L(\mathcal{H}_S\otimes\mathcal{H}_B)$. A general form for the interaction Hamiltonian that we assume is $H_I=\sum_{\alpha} A_{\alpha}\otimes B_{\alpha}$ with $A_{\alpha},B_{\alpha}$ self-adjoint operators and $\alpha$ runs over a finite set.
Clearly, as $\mathcal{B}(\mathcal{H}_S)$ is finite dimensional no problem on the operators domain arises.
More general cases where also $H_B$ or $H_I$ are time-dependent operators will not be treated in this work. We are interested in the reduced density matrix $\rho_S(t)=tr_B[\rho(t)]$, which can be obtained by using the projection operator approach introduced by Nakajima and Zwanzig \cite{Nakajima, Zwanzig} (see also \cite{Rivas_1} for more details).

We introduce two orthogonal projectors on $\mathcal{B}(\mathcal{H}_S\otimes\mathcal{H}_B)$, namely $\mathcal{P}\rho=tr_B[\rho]\otimes \rho_{th}$
and $\mathcal{Q}=\mathbb{I}-\mathcal{P}$, where $\rho_{th}=1/Z_B e^{-\beta H_B}$ is a thermal state for the environment; the approach works for any reference state $\rho_B$ for the environment, however one useful assumption 
that we make is $[\rho_B, H_B]=0$, hence for concreteness we will work assuming it is a Gibbs state. 
The equation, written in the interaction picture with respect to the free evolution $H_S(t)+H_B$ (henceforth indicated by the tilde), reads
\begin{equation}
\frac{d}{dt}\mathcal{P}\tilde{\rho}(t)=g^2\int\limits_0^t ds \mathcal{P}\mathcal{V}(t)\mathcal{G}(t,s)\mathcal{V}(s)\mathcal{P}\tilde{\rho}(s)
\end{equation}
where 
\begin{equation}
\begin{aligned}
\mathcal{V}(t)\rho &=-\imath [\tilde{H}_I(t),\rho], \\
\mathcal{G}(t,s) &=\mathcal{T}\exp\Big\{ g\int\limits_s^t dx\mathcal{Q}\mathcal{V}(x)\mathcal{Q}\Big\}, \\
\end{aligned}
\end{equation}
and $\tilde{H}_I(t)=U_0^{\dagger}(t,0)H_I U_0(t,0)$,
$U_0(t,0)=\big(\mathcal{T}\exp\{-\imath\int\limits_0^{t}H_S(s)ds\}\big) e^{-\imath H_B t}$. We remind the two assumptions underlying this equation:
\begin{enumerate}
\item $\mathcal{P}\mathcal{V}(t)\mathcal{P}=0$, which implies that $tr_B[\tilde{H}_I(t)\rho_B]=0$, a condition that can be always fulfilled
by modifying the Hamiltonian $H_S$.
\item $\rho(0)=\rho_S(0)\otimes \rho_{th}$, so there exists a (initial)
time such that the system and the environment are uncorrelated.
\end{enumerate}
For our derivation we make use of the integrated version
\begin{equation}
\begin{aligned}
\label{N_Z_int}
\mathcal{P}\tilde{\rho}(t)=&\mathcal{P}\tilde{\rho}(0)\\
&+g^2\int\limits_0^t ds\int\limits_0^s du\mathcal{P}\mathcal{V}(s)\mathcal{G}(s,u)\mathcal{V}(u)\mathcal{P}\tilde{\rho}(u).
\end{aligned}
\end{equation}

The aim of the present paper is to derive sufficient conditions such that the reduced density matrix $\tilde{\rho}_S(t)=tr_B[\tilde{\rho}(t)]$, described by Eq.\eqref{N_Z_int} under the assumptions listed above, evolves according to a Markovian quantum process characterised by a GKSL propagator.

As a matter of clarification, it is useful to briefly revise, in the case of time-independent Hamiltonians, the two main possible paths one can pursue:
\begin{enumerate}
\item A first approach, which is the standard microscopic derivation pursued in quantum optics (see for example \cite{Petruccione}), starts directly from the Von Neumann equation of system and environment and relies on a sequence of approximations/ansaetze, exemplified by the Born-Markov approximation. 
Underpinning this approach is the assumption of weak coupling to a sufficiently broadband reservoir so that the time evolution admits a 
separation of time scales, $\tau_S\gg\tau_B$, where $\tau_S$ is a characteristic time scale for $\tilde{\rho}_S(t)$ and $\tau_B$ is a decay time of some environment correlation functions. 
While many derivations are also based on a secular approximation (rotating-wave approximation) in order to get a GKLS generator, several recent works overcame this limitation (\cite{Wacker,Lidar,Davidovic, Nathan_1, Nathan_2, Konig}).
This derivation is physically intuitive and finds applications in a large range of phenomena involving atomic and molecular systems coupled to the radiation field.

\item The second approach is the derivation outlined by Davies in \cite{Davies_1,Davies_2}. The idea is to start from the Nakajima-Zwanzig Eq.\eqref{N_Z_int} and prove the uniform convergence of the density matrix in the weak-coupling limit $g\rightarrow 0$ and under a suitable time-rescaling:
\begin{equation}
\lim\limits_{g\rightarrow 0; t=\tau/g^2}\norm{\tilde{\rho}_S(t)-e^{\tau\mathcal{L}}}=0,
\end{equation}
where $\mathcal{L}$ is a time-independent GKLS generator.
In contrast to the previous approach, Davies' derivation does not need to explicitly introduce characteristic time scales.
Moreover, this derivation requires some regularity condition on the environment correlation functions, which is in a sense correspondent, albeit not equivalent, to the phenomenological Markov approximation mentioned above. Although an additional secular approximation is typically needed, later works suggested different procedures in order to avoid it (as for example in \cite{Taj}).
\end{enumerate}
Here we follow a protocol in the spirit of Davies, by considering a weak-coupling limit $g\rightarrow 0$ in the Eq.\eqref{N_Z_int}, after introducing a rescaled time $\tau=g^2 t$, and studying the uniform convergence of the reduced density matrix.
However, differently from the time-independent scenario, this rescaling affects also $H_S(t)$. As a consequence an additional hypothesis has to be made.

\begin{definition}
Let $X(a_1,a_2,...,a_N)$ be a family of elements in a given vector space $\mathcal{V}$, dependent on a finite number of real parameters $\{a_n\}_{n=1,...,N} $. Let us consider a real parameter $g$, a set of real indices $(x_1,x_2,...,x_N)$ and the scaling transformation $R_g: (a_1,...,a_N)\rightarrow (a_1 g^{x_1},...,a_N g^{x_N})$. We say that $X(a_1,a_2,...,a_N)$ is re-scalable (or invariant under rescaling) if $X(R_g(a_1,a_2,...,a_N))=X(a_1,a_2,...,a_N)$. 
\end{definition}
From the definition, we trivially notice that if $X(a_1,...,a_N)$ is re-scalable then $\lim\limits_{g\rightarrow 0}X(R_g(a_1,a_2,...,a_N))=X(a_1,a_2,...,a_N)$. Moreover, re-scalability holds also for the inverse transformation $X(R_g^{-1}(a_1,a_2,...,a_N))=X(a_1,a_2,...,a_N)$ 

In our work we will always assume invariance under rescaling of the system Hamiltonian $H_S(t,\lambda)$, where here we generally indicate with $\lambda$ an arbitrary set of parameters. As clarifying example, let us consider $H_S(t,v,h)=vt\sigma_z+ h\sigma_x$ for a two-level system, where $\sigma_z,\sigma_x$ are Pauli matrices; this Hamiltonian is re-scalable under any transformation of the type $R_g(t,v,h)=(t g^z, v g^{-z}, h)$, $\forall z$.
The purpose of the re-scalability is to reabsorb the $g$-dependence by renormalising the remaining set of parameters, every time we introduce a rescaled time $\tau=t g^z$. 

The example above indicates that this condition can always be achieved with a sufficient number of parameters and more importantly it preserves the form of the operator, as well as its spectral decomposition.
For the sake of notation, we shall indicate henceforth the set of parameters simply with $\lambda$ ($\lambda_R$ for the rescaled ones), therefore $H_S(t,\lambda)$ will denote the system Hamiltonian as dependent on time and a set of parameters, such that the set $\{t,\lambda\}$ satisfies the requirement of re-scalability.

Now we list some lemmas that will be useful for our derivation. They concern the estimation of oscillatory integrals  and the convergence of operators in the weak coupling limit. In particular, Theorem \ref{th_phase} and Lemma \ref{lemma_1} are already known results, while Lemma \ref{lemma_2} is novel.

\begin{theorem}{\textbf{(Method of stationary phase)}}.
\label{th_phase}
Given the function 
\begin{equation}
I(M)=\int\limits_a^b dx f(x)e^{\imath M\phi(x)}
\end{equation}
with $f,\phi$ real, analytic functions on $[a,b]$, the asymptotic limit $M\rightarrow\infty$ exhibits the following behaviour :
\begin{enumerate}
\item If $\phi$ has no stationary points in the interval $[a,b]$, namely $\partial_x\phi\neq 0$, $\forall x\in[a,b]$ 
\begin{equation}
\begin{aligned}
\label{cond_1}
I(M)\approx & \frac{f(b)}{\imath M \partial_x\phi(b)}e^{\imath M \phi(b)}  \\
&-\frac{f(a)}{\imath M \partial_x\phi(a)}e^{\imath M \phi(a)} .
\end{aligned}
\end{equation}
\item If $\partial_x\phi=0$ for $x_0\in (a,b)$:
\begin{equation}
\begin{aligned}
\label{cond_2}
I(M)\approx & f(x_0)\sqrt{\frac{2\pi}{ M |\partial_x^2\phi(x_0)|}} \\
& \times  e^{\imath M \phi(x_0)\pm\imath \partial_x^2\phi(x_0) \pi/4} .
\end{aligned}
\end{equation}
\item If $\partial_x\phi=0$ only for $x_0=a$ (analogous if we consider the upper limit)
\begin{equation}
\begin{aligned}
I(M)\approx &  \frac{f(b)}{\imath M \partial_x\phi(b)}e^{\imath M \phi(b)}+\frac{f(x_0)}{2}\\ & \times \sqrt{\frac{2\pi}{ M |\partial_x^2\phi(x_0)|}}
e^{\imath M \phi(x_0)\pm\imath \partial_x^2\phi(x_0) \pi/4} .
\end{aligned}
\end{equation}
\end{enumerate}
For a finite set of stationary points (including the boundaries) we can always divide the integration interval into a finite union of subintervals such as to apply one of these cases. 
\end{theorem}
If $f: I\subseteq \mathbb{R}\rightarrow \mathbb{C}$ analytic, then we just need to separate real and imaginary part in order to apply the theorem.
For a basic introduction and proof of this result, see for example \cite{Asymp}, whereas more details on error bounds can be found in \cite{Olver,Temme}.

\begin{lemma}
\label{lemma_1}
Let $\mathcal{B}$ be a $N$ dimensional Banach space with norm $\norm{.}$ and $\mathcal{B}^*$ with norm $\normsup{.}$.  Let $\{ T_g(\tau,\sigma) \}$ be a two-parameter family of bounded elements in $\mathcal{B}^*$ defined on a real compact interval $\tau\in [0,\tau^*]$ and $0\le g\le 1$. Let us consider the solution of the integral equation $y_g(\tau)=x+ \int\limits_0^{\tau}d\sigma T_g(\tau,\sigma)y_g(\sigma)= x+(\mathcal{T}_g y_g)(\tau)$  and assume that $\mathcal{T}_g$ converges strongly to $\mathcal{T}: \mathcal{B}\rightarrow \mathcal{B}$, which is bounded. 
Then $\forall x \in\mathcal{B}$,
\begin{equation}
\lim\limits_{g\rightarrow 0}\norm{y_g(\tau)-y(\tau)}=0
\end{equation}
uniformly in $[0,\tau^*]$, where $y(\tau)=x+(\mathcal{T}y)(\tau)$. 
\end{lemma}

\begin{proof}
Let us take $y_g(\tau)$ family of solutions of the integral equation $y_g(\tau)=x+(\mathcal{T}_g y_g)(\tau)$ for any $x\in\mathcal{B}$, we notice that this equation can be iterated by applying $\mathcal{T}_g$ to both members, obtaining
\begin{equation}
y_g(\tau)=x+(\mathcal{T}_g x)(\tau)+ (\mathcal{T}_g^2 x)(\tau)+...
\end{equation}
as well as for 
\begin{equation}
y(\tau)=x+ (\mathcal{T}x)(\tau)+ (\mathcal{T}^2 x)(\tau)+... \, .
\end{equation}
The maps $\mathcal{T}_g,\mathcal{T}$ are Volterra integral operators and therefore they are bounded in $[0,\tau^*]$ by hypothesis.
Hence we have
\begin{equation}
\begin{aligned}
\norm{y_g-y}&=\Big|\Big|\sum\limits_{n=1}^{\infty} \mathcal{T}_g^n x - \sum\limits_{n=1}^{\infty} \mathcal{T}^n x \Big|\Big|_1, \\
&\le \sum\limits_{n=1}^{\infty}\norm{\mathcal{T}_g^n x-\mathcal{T}^n x},
\end{aligned}
\end{equation}
but, for $g\rightarrow 0$, $\mathcal{T}_g x\rightarrow \mathcal{T}x$ because of the strong convergence in $[0,\tau^*]$, therefore $\lim\limits_{g\rightarrow 0}\norm{y_g-y}=0$.
\end{proof}

The content of this lemma can be also found in \cite{Davies_1}.

\begin{lemma}
\label{lemma_2}
Let $\mathcal{H}$ be a finite dimensional Hilbert space and $\mathcal{B(H)}$ the Banach space of linear operators on $\mathcal{H}$ equipped with the sup norm $\normsup{X}$. Let $H(t,\lambda)=\sum\limits_{n=1}^N E_n(t,\lambda)|n(t,\lambda)\rangle\langle n(t,\lambda)|$  be a Hamiltonian, which is assumed to be non-degenerate, differentiable in $t$ and re-scalable under the transformation $t=\tau/g^2,\lambda=\lambda_R/g^{\zeta}$ for some parameter $\zeta$.  Let $U(t,\lambda)=\mathcal{T}\exp\{ -\imath\int\limits_0^t ds H(s,\lambda)\}$  be the time-evolution operator, where $\mathcal{T}$ indicates the time-ordering. Then
\begin{equation}
\lim\limits_{g\rightarrow 0}\normsup{U(\tau/g^2,\lambda_R/g^{\zeta})-U_R(\tau,\lambda_R,g)}=0 ,
\end{equation}
for $0\le t g^2\le \tau^*$, where 
\begin{equation}
\begin{aligned}
U_R(\tau,\lambda_R, g)=\sum\limits_{n=1}^N \exp\Big\{ -\imath g^{-2}\int\limits_0^{\tau}d\sigma E_n(\sigma,\lambda_R)\\
-\imath\int\limits_0^{\tau}d\sigma \phi_n^B(\sigma,\lambda_R) \Big\} |n(\tau,\lambda_R)\rangle\langle n(0)|,
\end{aligned}
\end{equation} 
and $\phi_n^B(\tau,\lambda_R)=-\imath\langle n(\tau,\lambda_R)|\frac{d}{d\tau}|n(\tau,\lambda_R)\rangle$.
\end{lemma}


\begin{proof}
Let us consider the scaling transformation $t=\tau/g^2, \lambda=\lambda_R / g^{\zeta}$ which leaves invariant in form the Hamiltonian, i.e. $H(\tau/g^2,\lambda_R/g^{\zeta})=H(\tau,\lambda_R)$. We want to bound 
\begin{equation}
\normeu{U(\tau/g^2,\lambda_R/g^{\zeta})|\psi\rangle-U_R(\tau,\lambda_R,g)|\psi\rangle }, 
\end{equation}
where $|\psi\rangle=\sum\limits_n c_n(0)|n(0)\rangle \in\mathcal{H}$ is a generic vector. 
We can expand 
\begin{equation}
\begin{aligned}
U\Big(\frac{\tau}{g^2},\frac{\lambda_R}{g^{\zeta}}\Big)|\psi\rangle &=\sum\limits_{n=1}^N c_n\Big(\frac{\tau}{g^2},\frac{\lambda_R}{g^{\zeta}}\Big)|n(\tau,\lambda_R)\rangle, \\
U_R(\tau,\lambda_R,g)|\psi\rangle &=\sum\limits_{n=1}^N u_n(\tau,\lambda_R,g)c_n(0)|n(\tau,\lambda_R)\rangle,
\end{aligned}
\end{equation}
in terms of the instantaneous eigenbasis $|n(t,\lambda)\rangle$,
with $u_n(\tau,\lambda_R,g)=\exp\Big\{ -\imath g^{-2}\int\limits_0^{\tau}d\sigma E_n(\sigma,\lambda_R)-\imath\int\limits_0^{\tau}d\sigma \phi_n^B(\sigma,\lambda_R) \Big\}$ and in the first equation we used the re-scalability of the Hamiltonian eigenvectors.
The coefficients $c_n(t,\lambda)$ in $U(t,\lambda)|\psi\rangle=\sum\limits_n c_n(t,\lambda)|n(t,\lambda)\rangle$ are determined via
\begin{equation}
\begin{aligned}
\frac{d}{dt}|\psi(t)\rangle &=-\imath H(t,\lambda)|\psi(t)\rangle, \\
|\psi(0)\rangle &=|\psi\rangle, \\
\end{aligned}
\end{equation}
from which we obtain the linear differential equation
\begin{equation}
\label{eq_coeff_c}
\frac{d}{dt}c_n + \sum\limits_m^N c_m \langle n(t,\lambda)|\frac{d}{dt}|m(t,\lambda)\rangle=-\imath  E_n(t,\lambda) c_n ,
\end{equation}
whose solution can be expressed in the general form $c_n(t,\lambda)=\sum\limits_m^N T_{nm}(t,\lambda)c_m(0)$, with $T(t,\lambda)$ the $N\times N$ fundamental matrix.
Making use of the Cauchy-Schwartz inequality  $|\sum\limits_m^N x_m y_m|\le \sqrt{\sum\limits_m^N x_m^2}\sqrt{\sum\limits_m^N y_m^2}$ with
\begin{equation}
\begin{aligned}
x_m&=\sum\limits_n^N |T_{nm}\Big(\frac{\tau}{g^2},\frac{\lambda_R}{g^{\zeta}}\Big)-\delta_{nm}u_n(\tau,\lambda_R,g)|>0, \\
y_m&=|c_m(0)|>0 ,
\end{aligned}
\end{equation}
and that $ \sqrt{\sum\limits_m^N |c_m(0)|^2}\equiv ||\psi||=1$, we obtain the following chain of inequalities

\begin{widetext}
\begin{equation}
\label{ineq_1}
\begin{aligned}
\normeu{U\Big(\frac{\tau}{g^2},\frac{\lambda_R}{g^{\zeta}}\Big)|\psi\rangle-U_R(\tau,\lambda_R,g)|\psi\rangle} & \le  \sum\limits_n^N|c_n\Big(\frac{\tau}{g^2},\frac{\lambda_R}{g^{\zeta}}\Big)-u_n(\tau,\lambda_R,g)c_n(0)|\normeu{|n(\tau,\lambda_R)\rangle} , \\
&\le \sum\limits_{n,m}^N|T_{nm}\Big(\frac{\tau}{g^2},\frac{\lambda_R}{g^{\zeta}}\Big)-\delta_{nm}u_n(\tau,\lambda_R,g)||c_m(0)|, \\
&\le \sqrt{\sum\limits_{m}^N \Big|\sum\limits_n^N|T_{nm}(\tau/g^2,\lambda)-\delta_{nm}u_n(\tau,\lambda_R,g)|\Big|^2 } \sqrt{\sum\limits_m^N |c_m(0)|^2} , \\
&\le \sqrt{\sum\limits_{m}^N \Big|\sum\limits_n^N|T_{nm}(\tau/g^2,\lambda)-\delta_{nm}u_n(\tau,\lambda_R,g)|\Big|^2 } ||\psi||.
\end{aligned}
\end{equation}
\end{widetext}
Hence, we just need to show that $|T_{nm}\Big(\frac{\tau}{g^2},\frac{\lambda_R}{g^{\zeta}}\Big)-\delta_{nm}u_n(\tau,\lambda_R,g)|\rightarrow 0$ for $g\rightarrow 0$ in order to conclude the proof.
To this end, we consider the change of variable 
\begin{equation}
c_n\Big(\frac{\tau}{g^2},\frac{\lambda_R}{g^{\zeta}}\Big)=a_{n}(\tau,\lambda_R,g)u_n(\tau,\lambda_R,g)
\end{equation} 
and from Eq.\eqref{eq_coeff_c} we arrive at the integral equation
\begin{equation}
\label{eq_Fg_0}
\underline{a}_g(\tau,\lambda_R)=\underline{a}(0)+\int\limits_0^{\tau}d\sigma F_g(\sigma,\lambda_R)\underline{a}_g(\sigma,\lambda_R) ,
\end{equation}
where 
\begin{equation}
\underline{a}_g(\tau,\lambda_R)=(a_1(\tau,\lambda_R,g),...,a_N(\tau,\lambda_R,g))^T,
\end{equation}
\begin{equation}
\label{eq_Fg}
(F_g)_{nm}=
\begin{cases}
0 & \text{if } n=m, \\
- \alpha_{nm}e^{\imath g^{-2}\Delta_{nm}+\imath \eta_{nm}} & \text{if } n\neq m, \\ 
\end{cases}
\end{equation}
and
\begin{equation}
\begin{aligned}
\label{eq_coeff_F}
\alpha_{nm}(\tau,\lambda_R)&=\langle n(\tau,\lambda_R)|\frac{d}{d\tau}|m(\tau,\lambda_R)\rangle , \\
\Delta_{nm}(\tau,\lambda_R)&=\int\limits_0^{\tau}d\sigma \Big(E_n(\sigma,\lambda_R) -E_m(\sigma,\lambda_R)\Big) ,\\ 
\eta_{nm}(\tau,\lambda_R)&=\int\limits_0^{\tau}d\sigma \Big( \phi_n^B(\sigma,\lambda_R)-\phi_m^B(\sigma,\lambda_R)\Big) .\\ 
\end{aligned}
\end{equation}
$F_g(\tau,\lambda_R)$ is, by hypothesis, analytic in any compact interval $0\le\tau\le\tau^*$, thus applying Lemma \ref{lemma_1} we can see that the map
\begin{equation}
\underline{x}(\tau)\in\mathbb{C}^N\rightarrow \int\limits_0^{\tau}d\sigma F_g(\sigma,\lambda_R)\underline{x}(\sigma)\in\mathbb{C}^N 
\end{equation} 
converges strongly to zero  for $g\rightarrow 0$, regardless of the behaviour of the phases $\partial_{\sigma}\Delta_{nm}(\sigma,\lambda_R)=E_n(\sigma,\lambda_R)-E_m(\sigma,\lambda_R)$ in $\sigma\in [0,\tau]$ .
Hence, by Lemma \ref{lemma_2} we obtain that $\underline{a}_g(\tau,\lambda_R)\rightarrow \underline{a}(0)$ uniformly in $0\le\tau\le\tau^*$, $\forall \underline{a}(0)$, and therefore   
$|T_{nm}\Big(\frac{\tau}{g^2},\frac{\lambda_R}{g^{\zeta}}\Big)-\delta_{nm}u_n(\tau,\lambda_R,g)|\rightarrow 0$.

\end{proof}

\paragraph{Error bound in Lemma \ref{lemma_2}.}

For small but finite $g$, the error committed when replacing $U_S(\tau/g^2,\lambda_R/g^{\zeta})\rightarrow U_R(\tau,\lambda_R,g)$ can be bounded utilising Theorem \ref{th_phase} in the Eq. \eqref{eq_Fg_0},  with $\int\limits_0^{\tau}F_g(\tau,\lambda_R)\underline{a}(\sigma)\rightarrow 0$ for $g\rightarrow 0$ .

We divide the set of phases $\mathcal{F}=\{ \Delta_{n m}(\tau,\lambda_R) \}_{n\neq m}$ in Eq.\eqref{eq_coeff_F} as $\mathcal{F}=\mathcal{F}_{nc}\cup \mathcal{F}_c$, where
\begin{equation}
\begin{aligned}
\mathcal{F}_{nc}=\{ \Delta_{n m}\Big|  \nexists \sigma  \text{\space s.t.}  \frac{d \Delta_{n m}}{d\tau}(\sigma,\lambda_R)=0 \}, \\
\mathcal{F}_{c}=\{ \Delta_{n m} \Big|  \exists \sigma \text{\space s.t.}   \frac{d\Delta_{n m}}{d\tau}(\sigma,\lambda_R)=0 \}. \\
\end{aligned}
\end{equation}
In order to upper bound  $||\int\limits_0^{\tau}d\sigma F_g(\sigma,\lambda_R)\underline{a}(\sigma)||$,
for $g\rightarrow 0$, where the norm here is the Euclidean norm on $\mathbb{C}^N$, we consider $||\int\limits_0^{\tau}d\sigma F_g(\sigma,\lambda_R)\underline{a}(\sigma)||\le \sum\limits_{\Delta_{n m}\in \mathcal{F}_{nc} } \epsilon_{n m}^{n c} + \sum\limits_{\Delta_{n m}\in \mathcal{F}_{c} } \epsilon_{n m}^{c}$ and we make use of the expressions in Eq.\eqref{cond_1} and Eq.\eqref{cond_2}, to obtain, respectively,
\begin{equation}
\begin{aligned}
\epsilon_{n m}^{n c}\approx & g^2  \Big| \frac{\alpha_{n m}(\sigma,\lambda_R)}{\partial_{\sigma}\Delta_{n m}(\sigma,\lambda_R)} \Big| |a_m(\sigma,\lambda_R)|\Big]_{\sigma=0}^{\sigma=\tau} ,  \\
\epsilon_{n m}^{c}\approx  & g  \sqrt{\frac{2\pi}{|\partial_{\tau}^2\Delta_{n m}(\tau_0,\lambda_R) | } }  |\alpha_{n m}(\tau_0,\lambda_R)| \\
& |a_m(\tau_0,\lambda_R)|  ,  \\
\end{aligned}
\end{equation}
with $\tau_0$ such that $\partial_{\tau}\Delta_{n m}(\tau_0,\lambda_R)=0$ (assuming for simplicity that there exists only one point such that $\partial_{\tau}\Delta_{n m}(\tau_0,\lambda_R)=0$, however the extension to the general case is straightforward).
Therefore, using the explicit expressions of $\alpha_{n m}, \Delta_{n m}$ in Eq.\eqref{eq_coeff_F}, we can delineate a sufficient condition given by $\epsilon_{n m}^{n c},\epsilon_{n m}^{c}\ll 1$, namely
\begin{equation}
\begin{aligned}
\label{driving_cond_lemma}
\frac{\max\limits_{n\neq m; t} |\langle n(t)|\frac{d}{dt}|m(t)\rangle |}{\min\limits_{n\neq m; t}|E_n(t)-E_m(t)|}\ll g^{-2}, \text{\space if $\Delta_{n m} \in\mathcal{F}_{nc}$}, \\
\frac{\max\limits_{n\neq m; t_0} 2\pi |\langle n(t)|\frac{d}{dt}|m(t)\rangle |^2}{\min\limits_{n\neq m; t_0}|\dot{E}_n(t)-\dot{E}_m(t)|}\ll g^{-2}, \text{\space if $\Delta_{n m} \in\mathcal{F}_{c},$} \\
\end{aligned}
\end{equation}
with $t_0$ such that $E_n(t_0)=E_m(t_0)$ and we omitted for simplicity the dependence on $\lambda$. Notice that expressing the condition in Eq.\eqref{driving_cond_lemma} in terms of $t$ or $\tau$ is equivalent, given that $\tau$ plays the role of a dummy variable.

This lemma will be crucial in the forthcoming derivation of the time-dependent Markovian master equation, as it allows to replace the system time-evolution operator $U_S(t)$, which is generally not known or extremely complicated, with another one that can be entirely obtained from the eigendecomposition of $H_S(t)$, in the weak-coupling limit and under time rescaling. 

We are now ready to prove the main result of this section. 
\begin{theorem}
\label{th_main}
Let $\mathcal{B}(\mathcal{H}_S)$ be a finite dimensional Banach space for a system S, suppose that the dynamics of the quantum state $\tilde{\rho}_S(t)$ is given by the Nakajima-Zwanzig Eq.\eqref{N_Z_int}, with a bounded kernel, generated by the time-dependent re-scalable Hamiltonian $H(t)=H_S(t,\lambda)+H_B+g H_I$ in Eq.\eqref{general_h}, with $H_I=\sum\limits_{\alpha}A_{\alpha}\otimes B_{\alpha}$. Assume that $R_{\alpha,\beta}(x)=tr_B[e^{\imath H_B x}B_{\alpha} e^{-\imath H_B x} B_{\beta}\rho_{th}]$ are integrable functions on $\mathbb{R}^+$, i.e.
\begin{equation}
\label{regularity_cond}
\int\limits_0^{\infty}dx|R_{\alpha,\beta}(x)|\equiv C_{\alpha,\beta}< \infty .
\end{equation}
Then 
\begin{equation}
\lim\limits_{g\rightarrow 0}\Big|\Big|\tilde{\rho}_S(t)-\mathcal{T}\exp\Big\{ g^2\int\limits_0^{t}ds \tilde{\mathcal{L}}(s)\Big\} \rho_S(0)\Big|\Big|_1=0,
\end{equation}
for any finite time interval $0\le tg^2\le \tau^*$, where $\mathcal{\tilde{L}}(t)$ is a time-dependent GKLS generator.
\end{theorem}
\begin{proof}
Let us consider the scaling transformation $t=\tau/g^2,\lambda=\lambda_R/g^{\zeta}$, which leaves invariant the Hamiltonian $H_S(t,\lambda)$, with $\zeta$ depending on the explicit form of $H_S$. Let us introduce the operator $\chi_g(\tau,\lambda_R)\equiv \tilde{\rho}_S(\tau/g^2, \lambda_R/g^{\zeta})$, where $\tilde{\rho}_S(t,\lambda)$ is the reduced density matrix obtained from Eq.\eqref{N_Z_int}. We perform the change of variable $x = s-u $, $\sigma= u g^2 $ in the Nakajima-Zwanzig Eq.\eqref{N_Z_int}, to obtain 
\begin{equation}
\chi_g(\tau,\lambda_R)=\rho_S(0)+\int\limits_0^{\tau}d\sigma \mathcal{K}_g(\tau-\sigma,\sigma;\lambda_R)\chi_g(\sigma,\tau_R),
\end{equation}
where
\begin{widetext}
\begin{equation}
\label{kg_0}
\mathcal{K}_g(\tau-\sigma,\sigma;\lambda_R)\rho(\sigma)=\int\limits_0^{(\tau-\sigma)/g^{2}} dx tr_B\Big[\mathcal{P}\mathcal{V}(\sigma/g^{2}+x)\mathcal{G}(\sigma/g^{2}+x,\sigma/g^{2})\mathcal{V}(\sigma/g^{2})\mathcal{P}(\rho(\sigma)\otimes \rho_{th}) \Big],
\end{equation}
\end{widetext}
with $\mathcal{V}(\sigma/g^2)(*)\equiv -\imath[\tilde{H}_I(\sigma/g^2,\lambda_R/g^{\zeta}),*]$. We want to address the behaviour of the kernel for $g\rightarrow 0$. This is easily done by means of subsequent steps, as follows.

1. $\mathcal{G}(\sigma/g^{2}+x,\sigma/g^{2})\rightarrow \mathbb{I}$ strongly as $g\rightarrow 0$ , moreover plugging $\tilde{H}_I(t)=\sum\limits_{\alpha}\tilde{A}_{\alpha}(t)\otimes \tilde{B}_{\alpha}(t)$ and using the property $[H_B,\rho_{th}]=0$ we have
\begin{equation}
\begin{aligned}
\label{kg}
\mathcal{K}_g(\tau-\sigma,\sigma)\rho(\sigma)\rightarrow \sum\limits_{\alpha,\beta}\int\limits_0^{(\tau-\sigma)/g^{2}} dx R_{\alpha,\beta}(x)  \\ \times \Big[\tilde{A}_{\beta}(\sigma/g^2)\rho(\sigma),\tilde{A}_{\alpha}(\sigma/g^2+x)\Big] +h.c.,
\end{aligned}
\end{equation}
where $R_{\alpha,\beta}(x)=tr_B[\tilde{B}_{\alpha}(x)B_{\beta} \rho_{th}]$ are environment correlation functions and
\begin{equation}
\tilde{A}_{\alpha}\Big(\frac{\tau}{g^2}\Big)\equiv U_S^{\dagger}\Big(\frac{\tau}{g^2},\frac{\lambda_R}{g^{\zeta}}\Big)A_{\alpha} U_S\Big(\frac{\tau}{g^2},\frac{\lambda_R}{g^{\zeta}}\Big).
\end{equation}

2. Let $H_S(t)=\sum\limits_n E_n(t)|n_t\rangle\langle n_t|$ be the eigendecomposition of the system Hamiltonian, utilising Lemma \ref{lemma_2} and the  explicit form of $U_R(\tau, \lambda_R,g)$, we obtain for $g\rightarrow 0$
\begin{equation}
\begin{aligned}
\label{rel}
\tilde{A}_{\beta}\Big(\frac{\sigma}{g^2}\Big)\rightarrow  & \sum\limits_p e^{-\imath g^{-2} \Delta_p(\sigma,\lambda_R)}\\
&\times e^{-\imath\eta_p(\sigma,\lambda_R)}A_{\beta,p}(\sigma,\lambda_R)  ,\\
\tilde{A}_{\alpha}\Big(\frac{\sigma}{g^2}+x\Big)\rightarrow & \sum\limits_p e^{\imath g^{-2} \Delta_p(\sigma,\lambda_R)+\imath x\partial_{\sigma}\Delta_p(\sigma,\lambda_R)}\\
&\times e^{\imath\eta_p(\sigma,\lambda_R)}A_{\alpha,p}^{\dagger}(\sigma,\lambda_R) ,
\end{aligned}
\end{equation}
where we have introduced the double index $(n,m)\equiv p$ for the sake of simplicity and 
\begin{equation}
\begin{aligned}
A_{\alpha,p}(\tau,\lambda_R)\equiv &|m_0\rangle\langle m_{\tau}(\lambda_R)|A_{\alpha}|n_{\tau}(\lambda_R)\rangle\langle n_0|, \\
\Delta_{p}(\tau,\lambda_R)\equiv &\int\limits_0^{\tau}d\sigma(E_n(\sigma,\lambda_R)-E_m(\sigma,\lambda_R)).
\end{aligned}
\end{equation}
Notice that the decomposition in Eq.\eqref{rel} can be easily achieved in such a way that no degeneracies occur in the set $\{\Delta_p(\tau,\lambda_R)\}_p$. In the following, we will omit the dependence on the rescaled parameters $\lambda_R$ to ease the notation. Using these results, from Eq.\eqref{kg}
we get

\begin{equation}
\label{equiv_supermap}
\begin{aligned}
\mathcal{K}_g &(\tau-\sigma,\sigma)\rho(\sigma)\rightarrow  \\
& \sum\limits_{\alpha,\beta}\sum\limits_{p,q}\Big(\int\limits_0^{(\tau-\sigma)/g^{2}} dx R_{\alpha,\beta}(x) e^{\imath x\partial_{\sigma}\Delta_p(\sigma)} \Big)
\\ & \times e^{\imath g^{-2}(\Delta_p(\sigma)-\Delta_q(\sigma))}   e^{\imath (\eta_p(\sigma)-\eta_q(\sigma))} \\
& \times  \Big[ A_{\beta,q}(\sigma)\rho(\sigma),A_{\alpha,p}^{\dagger}(\sigma)\Big] +h.c.   .\\
\end{aligned}
\end{equation}

3. Applying Theorem \ref{th_phase}, the map $(\phi_g\rho)(\tau)\equiv\int\limits_0^{\tau}d\sigma \mathcal{K}_g(\tau-\sigma,\sigma)\rho(\sigma)$ converges strongly to 
$(\phi\rho)(\tau)=\sum\limits_{\alpha,\beta}\sum\limits_{p}\int\limits_0^{\tau} d\sigma \mathcal{\tilde{L}}(\sigma,\lambda_R)\rho(\sigma)$,
where
\begin{equation}
\begin{aligned}
\mathcal{\tilde{L}}(\tau,\lambda_R)\rho=& \sum\limits_{\alpha,\beta}\sum\limits_p \Gamma_{\alpha,\beta,p}(\tau)\\ 
&\times\Big[A_{\beta,p}(\tau)\rho, A_{\alpha,p}^{\dagger}(\tau )\Big] +h.c.
\end{aligned}
\end{equation}
and
\begin{equation}
\label{one-sidedF}
\Gamma_{\alpha,\beta,p}(\sigma)=\int\limits_0^{\infty} dx R_{\alpha,\beta}(x) e^{\imath x\partial_{\sigma} \Delta_p(\sigma)}
\end{equation}
is the one-sided Fourier transform of the correlation functions. Notice that $\phi$ is bounded, given the condition $ |\Gamma_{\alpha,\beta,k}(\sigma)|\le \int\limits_0^{\infty}dx|R_{\alpha,\beta}(x)|< \infty$ by hypothesis.

In fact, for the contributions with $q\neq p$  in Eq.\eqref{equiv_supermap} we can individuate the phases (once again, we omit here the dependence on the rescaled parameters $\lambda_R$) $\varphi_{p q}(\tau)=\Delta_p(\tau)-\Delta_q(\tau)$ and  $M=g^{-2}$ defined in Theorem \ref{th_phase}. For $g\rightarrow 0$, we can also provide an estimate of the error committed, depending on the behaviour of the phases $\varphi_{p q}$.
In particular, let
\begin{equation}
\mathcal{F}=\{\varphi_{p q}\}=\mathcal{F}_{nc} \cup \mathcal{F}_c
\end{equation}
be the set of phases for $p\neq q$ that enter the master equation, which we further divide into the union of two sets as
\begin{equation}
\begin{aligned}
\mathcal{F}_{nc}=\{ \varphi_{p q}\Big|  \nexists \sigma  \text{\space s.t.}  \frac{d\varphi_{p q}}{d\tau}(\sigma)=0 \}, \\
\mathcal{F}_{c}=\{ \varphi_{p q}\Big|  \exists \sigma \text{\space s.t.}   \frac{d\varphi_{p q}}{d\tau}(\sigma)=0 \}. \\
\end{aligned}
\end{equation}
From Eq.\eqref{equiv_supermap} we can derive the upper bound
\begin{equation}
\begin{aligned}
&\norm{(\phi_g - \phi)\rho}\le  2 \sum\limits_{\alpha,\beta}\sum\limits_{q\neq p}\sum\limits_j  \norm{F_j} \\
& \times \Big| \int\limits_0^{\tau}d\sigma \Gamma_{\alpha,\beta,p}(\sigma) e^{\imath g^{-2}\varphi_{p q}(\sigma)} \\
& \times e^{\imath (\eta_p(\sigma)-\eta_q(\sigma))} f_{\alpha,\beta,k,p}^j(\sigma)| \Big|,
\end{aligned}
\end{equation}
where we have expressed the operator $[A_{\beta,q}(\sigma)\rho(\sigma),A_{\alpha,p}^{\dagger}(\sigma)]=\sum\limits_j f_{\alpha,\beta,q,p}^j(\sigma) F_j$ in terms of the canonical basis $\{F_j \}_{j}$ of $\mathcal{B}(\mathcal{H}_S)$ and the factor $2$ comes from the hermitian conjugate contribution.
Now, we distinguish among the terms whose phase belongs to either $\mathcal{F}_{nc}$ or $\mathcal{F}_c$, so that $\norm{(\phi_g - \phi)\rho}\le \sum\limits_{\varphi_{pq}\in\mathcal{F}_{nc}}E_{p q}^{nc}+\sum\limits_{\varphi_{pq}\in\mathcal{F}_{c}}E_{p q}^{c}$.
In particular for the first set we can apply the asymptotic expression in  Eq.\eqref{cond_1} to obtain
\begin{equation}
\begin{aligned}
E_{p q}^{nc}\approx \sum\limits_{\alpha,\beta} \sum\limits_j g^2 \norm{F_j}
\Big[\Big|\frac{\Gamma_{\alpha,\beta,p}(\sigma)}{\partial_{\sigma}\varphi_{p q}(\sigma)} \Big|\\
\times |f_{\alpha,\beta,p,q}^j(\sigma)| \Big]_{\sigma=0}^{\sigma=\tau},
\end{aligned}
\end{equation}
while for the second one the expression we can safely use is Eq.\eqref{cond_2}
\begin{equation}
\begin{aligned}
E_{p q}^c\approx \sum\limits_{\alpha,\beta}\sum\limits_j||F_j||_1 \sqrt{\frac{2\pi g^2}{|\partial_{\tau}^2 \varphi_{p,q}(\tau_0)|}} \\
\times \Big|\Gamma_{\alpha,\beta;p}(\tau_0)f^j_{\alpha,\beta,p,q}(\tau_0) \Big| ,
\end{aligned}
\end{equation}
where we have assumed without any loss of generality that only one critical point $\tau_0$, such that $\partial_{\tau}\varphi_{pq}(\tau_0)=0$, is present inside the time interval of interest.

The error due to this approximation vanishes in the limit $g\rightarrow 0$ provided that
\begin{equation}
\label{bound_secapprox}
\begin{aligned}
g^2 \Big|\frac{\Gamma_{\alpha,\beta,p}(\tau)}{\partial_{\tau}\varphi_{p q}(\tau)} \Big| &\ll 1 \text{\space\space if $\varphi_{p q}\in\mathcal{F}_{nc} $ }, \\
g\Big|\frac{\Gamma_{\alpha,\beta;p}(\tau_0)}{\sqrt{|\partial_{\tau}^2 \varphi_{p q}(\tau_0)| }} \Big|& \ll 1  \text{\space\space if $\varphi_{p q}\in\mathcal{F}_{c} $ },\\
\end{aligned}
\end{equation}
$\forall \alpha,\beta, q\neq p, \forall \tau_0$ s.t. $\partial_{\tau}\varphi_{p q}(\tau_0)=0$. In particular, using that $|\Gamma_{\alpha,\beta,p}(\tau)|\le C_{\alpha,\beta}$ and indicating with $ \Omega_p(\tau)=\partial_{\tau}\Delta_p(\tau)$ the instantaneous Bohr frequencies which enter the master equation, we arrive at the sufficient condition
\begin{equation}
\begin{aligned}
\label{general_secular_th}
\frac{\min\limits_{p\neq q; t} |\Omega_p(t)-\Omega_q(t)|}{\max\limits_{\alpha,\beta} C_{\alpha,\beta}} \gg g^2, \text{\space\space if $\varphi_{p q}\in\mathcal{F}_{nc} $ } \\
\frac{\min\limits_{p\neq q; t^{*}} |\dot{\Omega}_p(t^{*})-\dot{\Omega}_q(t^{*})|}{\max\limits_{\alpha,\beta}C_{\alpha,\beta}^2}  \gg g^2, \text{\space\space if $\varphi_{p q}\in\mathcal{F}_{c} $ }
\end{aligned}
\end{equation}
where $t^*$ indicates the set of critical points such that $\Omega_p(t^{*})=\Omega_q(t^{*})$ and the dot represents the time derivative. Notice that expressing the condition in Eq.\eqref{general_secular} in terms of $t$ or $\tau$ is equivalent, given that $\tau$ plays the role of a dummy variable.

As a consequence, putting all together we can apply Lemma \ref{lemma_1} in a given interval $[0,\tau^*]$, in order to prove the uniform convergence 
\begin{equation}
\lim\limits_{g\rightarrow 0}\Big|\Big|\chi_g(\tau,\lambda_R)-\chi(\tau,\lambda_R)\Big|\Big|_1=0,
\end{equation}
where $\chi(\tau,\lambda_R)$ is solution of the integral equation $\chi(\tau,\lambda_R)=\rho_S(0)+\int\limits_0^{\tau}d\sigma\mathcal{\tilde{L}}(\sigma,\lambda_R)\chi(\sigma,\lambda_R)$, or equivalently
\begin{equation}
\chi(\tau,\lambda_R)=\mathcal{T}\exp\Big\{ \int\limits_0^{\tau}d\sigma\mathcal{\tilde{L}}(\sigma,\lambda_R)\Big\} \rho_S(0).
\end{equation}
One can easily check that by construction $\mathcal{\tilde{L}}(\tau,\lambda_R)=\mathcal{\tilde{L}}(t g^2,\lambda g^{\zeta})=\mathcal{\tilde{L}}(t,\lambda)$.
Therefore, because of $\chi_g(\tau,\lambda_R)=\tilde{\rho}_S(\tau/g^2, \lambda_R/g^{\zeta})$, substituting back the transformation $t=\tau/g^2, \lambda=\lambda_R/g^{\zeta}$ we obtain the uniform convergence
\begin{equation}
\begin{aligned}
\tilde{\rho}_S(t,\lambda)&\rightarrow \mathcal{T}\exp\Big\{ \int\limits_0^{t g^2}d\sigma\mathcal{\tilde{L}}(\sigma,\lambda g^{\zeta})\Big \} \rho_S(0) \\
&=\mathcal{T}\exp\Big\{ g^2\int\limits_0^{t}ds\mathcal{\tilde{L}}(s g^2,\lambda g^{\zeta})\Big\} \rho_S(0) \\
&=\mathcal{T}\exp\Big\{ g^2\int\limits_0^{t}ds\mathcal{\tilde{L}}(s,\lambda)\Big\} \rho_S(0).
\end{aligned}
\end{equation}
The last thing we need to prove is that  $\mathcal{\tilde{L}}(t)$ is indeed a GKLS generator. To this end, we split $\Gamma_{\alpha,\beta,p}(t)= \frac{1}{2}\gamma_{\alpha,\beta,p}(t)+\imath S_{\alpha,\beta,p}(t)$, where $\gamma_{\alpha,\beta,p}(t)$ and $S_{\alpha,\beta,p}(t)$ form Hermitian matrices in the indices $\alpha,\beta$, ending up with 
\begin{equation}
\begin{aligned}
\tilde{\mathcal{L}}(t)\rho&=-\imath[\tilde{H}_{LS}(t),\rho]+
\sum\limits_{\alpha,\beta,p}\gamma_{\alpha,\beta,p}(t)\Big(A_{\beta,p}(t) \\
&\times \rho A_{\alpha,p}^{\dagger}(t)-\frac{1}{2}\{  A_{\alpha,p}^{\dagger}(t)A_{\beta,p}(t), \rho\} \Big),
\end{aligned}
\end{equation}
with $\tilde{H}_{LS}=\sum\limits_{\alpha,\beta; p}S_{\alpha,\beta,p}(t) A_{\alpha,p}^{\dagger}(t)A_{\beta,p}(t)$ self-adjoint. Diagonalisation of the matrices $\gamma_{p}(t)$, with entries $\gamma_{\alpha,\beta,p}(t)$,  brings the generator $\mathcal{\tilde{L}}(t)$ in the form of Eq.\eqref{gkls_gen}. 
The positivity of the dissipation rates follows from the fact that $\gamma_{p}(t)$ is actually positive semidefinite. Given a generic  $\underline{v}\in\mathbb{C}^N$ and using the explicit form of $\gamma_{\alpha,\beta,p}(t)=\int\limits_{-\infty}^{\infty} dx e^{\imath \partial_{t} \Delta_p(t) x} R_{\alpha,\beta}(x) $, we have
\begin{equation}
(\underline{v},\gamma_p(t)\underline{v})=\int\limits_{-\infty}^{\infty} dx  e^{\imath \partial_{t} \Delta_p(t) x} R(x),
\end{equation}
with $R(x)=\sum\limits_{\alpha,\beta}v_{\alpha}^{*} R_{\alpha,\beta}(x) v_{\beta}$. For any fixed $t$, the scalar product $(\underline{v},\gamma_p(t)\underline{v})$ is nothing but the Fourier transform of $R(x)$, thus by using Bochner's theorem (see for example \cite{Rivas_1}) one can easily see that the Fourier transform of $R(x)$ is a positive function. As a consequence, $\gamma_p(t)$ is a positive semidefinite matrix $\forall t,p$.
\end{proof}

Now that the convergence of the reduced density matrix in the interaction picture is demonstrated, we can obtain the evolution in the Schrödinger picture using the relation $\rho_S(t)=\mathcal{U}_{t,0}\tilde{\rho}_S(t)=U_S(t,0) \tilde{\rho}_S(t)U_S^{\dagger}(t,0)$, with $U_S(t,0)=\mathcal{T}\exp\{ -\imath\int\limits_0^t ds H_S(s)\}$.

\paragraph{Summary of the main results.}
We end this section with a self-contained summary of the novel results obtained. Let $H(t)=H_S(t)+H_B+g H_I$ be the Hamiltonian of system and environment as in Eq.\eqref{general_h}, where $H_S(t)=\sum\limits_n E_n(t)|n_t\rangle\langle n_t|$, $H_I=\sum\limits_{\alpha} A_{\alpha}\otimes B_{\alpha}$, while let $U_S(t)=\mathcal{T}\exp\{\int\limits_0^t ds H_S(s)\}$ be the system time-evolution operator.
The initial state of the environment will be called $\rho_{th}$ and must satisfy $[H_B,\rho_{th}]=0$, such as a thermal state.

In the weak-coupling limit, precisely determined by the conditions discussed in the following points, the exact dynamics of the reduced density matrix $\rho_S(t)$ can be described, in the Schrödinger picture, by the Markovian master equation 
\begin{equation}
\begin{aligned}
\label{time-dep-Markov}
\frac{d}{dt}\rho_S=&-\imath[H_S(t)+g^2 H_{LS}(t),\rho_S] \\
&+ g^2\sum\limits_{\alpha,\beta}\sum\limits_p \gamma_{\alpha,\beta, p}(t)\Big( L_{\beta,p}(t)\rho_S L_{\alpha,p}^{\dagger}(t)\\
&-\frac{1}{2}\{ L_{\alpha,p}^{\dagger}(t)L_{\beta,p}(t),\rho_S \}    \Big),
\end{aligned}
\end{equation}
where the jump operators and the so-called Lamb shift Hamiltonian are
\begin{equation}
\begin{aligned}
L_{\alpha, p}(t)&\stackrel{p=(n,m)}{=}U_S(t)|m_0\rangle\langle m_t|A_{\alpha}|n_t\rangle\langle n_0|U_S^{\dagger}(t), \\
H_{LS}(t)&=\sum\limits_{\alpha,\beta}\sum\limits_p S_{\alpha,\beta,p}(t) L_{\alpha,p}^{\dagger}(t)L_{\beta,p}(t),
\end{aligned}
\end{equation}
whereas
\begin{equation}
\begin{aligned}
\label{general_one_sided}
\int\limits_0^{\infty} dx R_{\alpha,\beta}(x) e^{\imath x(E_n(t)-E_m(t))}\stackrel{p=(n,m)}{=}\\
=\frac{1}{2}\gamma_{\alpha,\beta,p}(t)+\imath S_{\alpha,\beta,p}(t),
\end{aligned}
\end{equation}
and $R_{\alpha,\beta}(t)=tr_B[e^{\imath H_B t}B_{\alpha} e^{-\imath H_B t} B_{\beta}\rho_{th}]$ are the environment correlation functions.

Finally, we highlight some comments and a discussion on the range of validity of the Eq.\eqref{time-dep-Markov}.

\begin{enumerate}

\item
The jump operators in the Schrödinger picture are neither eigenoperators of the time-evolution operator $U_S(t)$ nor of $H_S(t)$, differently from the standard derivation in the time independent scenario.
Consequently, the Lamb shift Hamiltonian does not generally
commute with $H_S(t)$. As a result, there is a unitary correction to the time evolution of the closed system that can induce coherent transitions between the instantaneous energy levels of $H_S(t)$.

\item As already pointed out, Theorem \ref{th_main} provides sufficient conditions such that the solution 
of the Nakajima-Zwanzig Eq.\eqref{N_Z_int} can be approximated with a Markovian evolution.
For finite but small $g$, this approximation is clearly affected by some errors.
A first aspect to point out is given by the fact that our proof, done employing the rescaled time $\tau=t g^2$, guarantees the uniform convergence inside a given interval $0\le \tau\le \tau^*$. For finite $g$, this implies an upper bound to the domain of convergence given by $\tau^* g^{-2}$. However, this problem can be easily overcome if we assume that $\tau^*$ can be taken arbitrary large (see also \cite{Merkli}, regarding validity of Markovian master equations for all time scales).

\item \textit{Markov approximation}. The condition of regularity on the environment correlation functions stated in Eq.\eqref{regularity_cond}, namely $\int\limits_0^{\infty}dt|R_{\alpha,\beta}(t)|\equiv C_{\alpha,\beta}< \infty$, adopts the role of the 
so-called Markov approximation in the limit $g\rightarrow 0$. 
Indeed, it is clear that only for a sufficiently fast decay  in the long-time limit, such as $|R_{\alpha,\beta}(t)|=|tr_B[e^{\imath H_B t}B_{\alpha} e^{-\imath H_B t} B_{\beta}\rho_{th}]|\sim 1/t^{a} $ with $a>1$, the convergence can be guaranteed.

Furthermore, we notice that if $H_B$ has a discrete spectrum then $R_{\alpha,\beta}(t)$ is a periodic function, therefore the only possibility to satisfy this condition is to require that the environment is characterised by a continuous spectrum (see \cite{Rivas_1}).

\item \textit{Secular approximation}. 

This  sufficient  condition is derived from the bound given in Eq.\eqref{bound_secapprox} and Eq.\eqref{general_secular_th}. Let $\{\Omega_p(t)\}_p$ be the set of instantaneous Bohr frequencies of $H_S(t)$, with $\Omega_p(t)= \Omega_q(t), \forall t$, if and only if $p=q$. Moreover, let us indicate with $\mathcal{D}=\{ f:\mathbb{R}\rightarrow \mathbb{R}| \exists\{t_n\}_n \text{\space s.t\space} f(t_n)=0 \}$ the set of real functions which are zero in at most a countable number of points. The condition reads  
\begin{equation}
\begin{aligned}
\label{general_secular}
\frac{\min\limits_{p\neq q; t} |\Omega_p(t)-\Omega_q(t)|}{\max\limits_{\alpha,\beta} C_{\alpha,\beta}} & \gg  g^2, \\
& \text{\space\space if  $\Omega_{p}- \Omega_{q}\notin\mathcal{D} $,} \\
\frac{\min\limits_{p\neq q; t^*} |\dot{\Omega}_p(t^*)-\dot{\Omega}_q(t^*)|}{\max\limits_{\alpha,\beta}C_{\alpha,\beta}^2} & \gg g^2, \\
& \text{\space\space if $\Omega_{p}- \Omega_{q}\in\mathcal{D} $, }
\end{aligned}
\end{equation}
where in the second inequality the set of points $t^*$ are such that $\Omega_p(t^*)-\Omega_q(t^*)=0$, whereas the dot indicates the time derivative.

From a physical view point, this requirement implies that the frequencies related to the transitions between energy levels are sufficiently well spaced, during the entire time interval of interest. In the case of constant Hamiltonian, only the first condition in Eq.\eqref{general_secular} holds and reproduces the standard rotating-wave approximation employed in time-independent Markovian master equations (\cite{Petruccione,Rivas_1}), whereas the second condition is completely novel compared to other approaches and becomes particularly relevant in case at certain times degeneracies appear in the set of Bohr frequencies.
As a consequence, naive applications of rotating-wave-like approximations, which are based on the simple condition of large differences $|\Omega_p(t)-\Omega_q(t)|$, are not able to provide this kind of condition, which stems uniquely from the application of the method of stationary phase.

\item \textit{Condition on the driving}. 

A final consideration concerns Lemma \ref{lemma_2} and the condition derived in Eq.\eqref{driving_cond_lemma}, which corresponds to a constraint on the driving that naturally emerges in our derivation, as result of the combination of time-rescaling and the limit $g\rightarrow 0$. Let  $\mathcal{D}=\{ f:\mathbb{R}\rightarrow \mathbb{R}| \exists\{t_n\}_n \text{\space s.t\space} f(t_n)=0 \}$ be the set defined previously, the condition reads
\begin{equation}
\begin{aligned}
\label{driving_cond}
\frac{\max\limits_{n\neq m; t} |\langle n_t|\frac{d}{dt}|m_t\rangle |}{\min\limits_{n\neq m; t}|E_n(t)-E_m(t)|} & \ll g^{-2}, \\
& \text{\space if $E_n-E_m \notin\mathcal{D}$}, \\
\frac{\max\limits_{n\neq m; t^{*}} 2\pi |\langle n_t|\frac{d}{dt}|m_t\rangle(t^{*}) |^2}{\min\limits_{n\neq m; t^{*}}|\dot{E}_n(t^{*})-\dot{E}_m(t^{*})|} & \ll g^{-2},  \\
& \text{\space if $E_n-E_m  \in\mathcal{D},$} \\
\end{aligned}
\end{equation}
with $t^{*}$ such that $E_n(t^{*})-E_m(t^{*})=0$.

Interestingly, $U_R$ defined in Lemma \ref{lemma_2} has the same structure of the so called adiabatic time-evolution operator, which takes the form
\begin{equation}
\label{ad_op_gen}
U_S^{ad}(t)=\sum\limits_{n} e^{-\imath\int\limits_0^t ds (E_n(s)+\phi_n^B(s) )} |n_t\rangle\langle n_0|,
\end{equation}
with $\phi_n^B(t)=\langle n_t|\frac{d}{dt}|n_t\rangle$. In the standard formulation of the adiabatic theorem
(\cite{Sakurai}), the operator in Eq.\eqref{ad_op_gen} approximates the time-evolution operator if the so called adiabatic condition, namely
\begin{equation}
\label{adiabatic_cond}
\frac{\max\limits_{n\neq m; t\in[0,t^*]} |\langle n_t|\frac{d}{dt}|m_t\rangle |}{\min\limits_{n\neq m; t\in[0,t^*]}|E_n(t)-E_m(t)|}\ll 1  ,
\end{equation}
is satisfied, within the interval of interest $[0,t^*]$. Given that $g\ll 1$, the inequality \eqref{driving_cond} clearly extends way beyond the constraint imposed by adiabaticity, allowing us to explore regimes of strong driving. Notice that the condition in Eq.\eqref{adiabatic_cond} also expresses the  basic assumption underpinning the derivation of the adiabatic master equation \cite{Zanardi}. In light of this consideration, despite the use of the adiabatic time-evolution operator as a consequence of Lemma \ref{lemma_2}, we underline that Eq.\eqref{time-dep-Markov} is a good candidate to describe driven systems outside the adiabatic regime.
Naturally, our master equation reproduces by construction the adiabatic master equation, an aspect that will be analysed more in detail in the next section.

\end{enumerate}

\section{First example: periodically driven qubit system}

To illustrate the performance of our derived master equation, we consider the specific case of periodic driving on a quantum two-level system (\emph{qubit}) \cite{Grossmann, Grifoni}.
Originally, this was also the first example of a rigorous derivation of time-dependent Markovian master equations that could be used to analyse strongly driven systems using Floquet theory
\cite{Chu_1,Hanggi_1,Breuer,Chu_2,Alicki_3, Floquet,Hartmann}.
Indeed, when it is possible to analytically find the so called Floquet representation of the time-evolution operator $U_S(t)$, a Markovian master equation can be derived by using the very same procedure as in the time-independent scenario. However, for many interesting and non-trivial cases finding this representation is a very hard task and solutions are known only for a few special cases. 
We also point out that recent examples of Markovian master equations beyond the adiabatic regime rely on the Floquet theory when describing periodic drivings (\cite{Dann}).
Our approach took a different path altogether, offering an alternative and feasible method for addressing strong periodic drivings.
In order to test our approach, we now study and benchmark against tensor network simulations a model whose Floquet representation is not known analytically, namely a time-dependent dynamics of the form,
\begin{equation}
\begin{aligned}
H_S(t)&=\left(
\begin{array}{cc}
\omega_0 & \Omega\sin(\omega t) \\
\Omega\sin(\omega t) & -\omega_0 \\
\end{array}
\right), \\
&\equiv \omega_0\sigma_z + \Omega\sin(\omega t)\sigma_x ,
\end{aligned}
\end{equation}
where $\sigma_x,\sigma_z$ are the usual Pauli matrices and $\omega_0$, $\omega$, $\Omega$ are real, positive parameters denoting the system's natural frequency, the frequency of the driving field and the so-called Rabi frequency, a measure of the driving strength, respectively. 

The driven system is coupled bilinearly to a bosonic environment consisting of independent harmonic oscillators and described by the total Hamiltonian (in natural units $\hbar=1$):
\begin{equation}
H=H_S(t)+\sum\limits_j w_j b^{\dagger}_j b_j +\sum\limits_j \sigma_x \otimes g_j(b_j+b^{\dagger}_j),
\end{equation}
where the bosonic creation/annihilation operators satisfy the usual canonical commutation relations $[b_i,b^{\dagger}_j]=\delta_{ij}$, $[b_i,b_j]=0$  and $w_j>0$ without loss of generality. 
For our analysis, system and environment are initially prepared in the thermal states $\rho_S(0)=e^{-H_S(0)/T_S}/Z_S$ and $\rho_{th}=e^{-H_B/T_B}/Z_B$, with $T_S,T_B$ denoting, respectively, the system and environmental temperature.
The coefficients $g_j$ are related to the environment spectral density $J(w)=\sum_j g_j^2\delta(w-w_j)$ via the identity $\sum_j g_j^2=\int_0^{\infty}dw J(w)$. In the continuum limit, we will consider specifically an Ohmic spectral density with exponential cutoff as $J(w)=a we^{-w/w_c}$ with $a,w_c>0$.

The Markovian master equation for the driven two-level system, derived from Eq.\eqref{time-dep-Markov}, reads
\begin{equation}
\begin{aligned}
\label{generic_TDME}
\frac{d}{dt}\rho_S=&-\imath[H_S(t)+H_{LS}(t),\rho_S] \\
&+\sum\limits_{p=0,\pm}\gamma_p(t)\mathcal{D}[L_p(t)]\rho_S,
\end{aligned}
\end{equation}
where the dissipator is of Lindblad form, namely $\mathcal{D}[L_p]\rho=L_p \rho L_p^{\dagger}-\frac{1}{2}\{L_p^{\dagger}L_p , \rho \}$, the time-dependent jump operators expressed in the Schrödinger picture are given by
\begin{equation}
\label{jump_periodic}
\begin{aligned}
L_0(t)&=\sin\varphi\Big( \frac{1}{2}(|\alpha|^2-|\beta|^2)\sigma_z-2\alpha\beta\sigma_+ + h.c. \Big), \\
L_+(t)&=\cos\varphi\Big( \alpha\beta^*\sigma_z + \alpha^2\sigma_+ -\beta^{* 2}\sigma_-  \Big), \\
L_-(t)&=\cos\varphi\Big( \alpha^*\beta\sigma_z + \alpha^{* 2}\sigma_- -\beta^{ 2}\sigma_+  \Big), \\
\end{aligned}
\end{equation}
with $\tan\varphi(t)=\Omega\sin(\omega t)/\omega_0$, while the complex functions $\alpha(t),\beta(t)$ constitute the entries of the time-evolution operator
\begin{equation}
U_S(t)=\left(
\begin{array}{cc}
\alpha(t)& \beta(t) \\
-\beta^*(t) & \alpha^*(t) \\
\end{array}
\right) 
\end{equation}
and are determined by the solution of the system of coupled differential equations
\begin{equation}
\label{system_param}
\begin{aligned}
\imath\frac{d}{dt}\alpha&=\omega_0\alpha-\Omega\sin(\omega t)\beta^* , \\
\imath\frac{d}{dt}\beta^*&=-\omega_0\beta^*-\Omega\sin(\omega t)\alpha, \\
\end{aligned}
\end{equation}
with initial conditions $\alpha(0)=1,\beta(0)=0$.
The corresponding dissipation rates are given by
\begin{equation}
\label{eq_rates}
\begin{aligned}
\gamma_0(t)&=4\pi a T_B , \\
\gamma_+(t)&=2\pi J(2 E(t))\bar{n}_{BE}(2E(t);T_B) ,\\
\gamma_-(t)&=2\pi J(2E(t))(1+\bar{n}_{BE}(2E(t);T_B)) , \\
\end{aligned}
\end{equation}
where $E(t)=\sqrt{\omega_0^2+\Omega^2\sin^2(\omega t)}$, $\bar{n}_{BE}(w,T_B)=1/(e^{w/T_B}-1)$ is the Bose-Einstein distribution and the Lamb shift Hamiltonian
\begin{equation}
\begin{aligned}
H_{LS}(t)=&-\frac{1}{2}S(t)\cos^2\varphi\Big( (|\alpha|^2-|\beta|^2)\sigma_z \\
&-2\alpha\beta\sigma_+ -2\alpha^*\beta^*\sigma_-  \Big),
\end{aligned}
\end{equation}
with
\begin{equation}
\begin{aligned}
\label{S_t}
S(t)=& \int\limits_0^{\infty}dw J(w)(1+2\bar{n}_{BE}(w))\\
&\times \Big[ \mathcal{P}\Big( \frac{1}{w+2E(t)}\Big) - \mathcal{P}\Big( \frac{1}{w-2E(t)}\Big) \Big].
\end{aligned}
\end{equation}
Note that $H_{LS}(t)$ does not generally commute with the system Hamiltonian. Specific details of the derivation can be found in Appendix A. We notice that, despite the general structure in Eq.\eqref{jump_periodic}, the  explicit form of $\alpha(t),\beta(t)$ is not known in general and have to be computed numerically. Hence, one should in general rely on a semi-analytical approach, where the jump operators can be only parametrised in terms of $\alpha,\beta$. 

The validity of the master equation is determined by the following conditions.
\begin{enumerate}
\item \textit{Weak-coupling limit}. The strength of the interaction between the system and the environment is determined by the set of coefficients $g_j$, or in the continuum limit by the spectral density $J(w)$. Notice that, in units of $\hbar=1$, the ohmic spectral density $J(w)=awe^{-w/w_c}$ has the dimension of an energy, therefore $a$ is an adimensional parameter that we can conveniently identify with $g^2$, according to the identity $\sum_j g_j^2=\int_0^{\infty}dw J(w)$. In the parameter regime considered here, a typical requirement of weak coupling corresponds to $g \sim 10^{-2}$.

\item \textit{Markov approximation}.  The environment correlation function $R(t)=tr_B[e^{\imath H_B t}Be^{-\imath H_B t}  B\rho_{th}]$ needs to satisfy the condition of integrability, i.e. $\int\limits_0^{\infty}dt |R(t)|\equiv C<\infty$ (Eq.\eqref{regularity_cond}). Given that $g\equiv \sqrt{a}$, the operator $B$ in the interaction Hamiltonian in Eq.\eqref{general_h} can be identified as $B=\frac{1}{\sqrt{a}}\sum\limits_j g_j (b_j+b_j^{\dagger})$.

The proof can be found in the Appendix C.

\item  \textit{Secular approximation}. This is determined by the requirement in Eq.\eqref{general_secular}, which reads here 
\begin{equation}
\min_{p\neq q;t}|\Omega_p(t)-\Omega_q(t)|=2\omega_0 \gg a C,
\end{equation}
with $\{0,\pm 2E(t) \} $ the set of Bohr frequencies. 

\item \textit{Condition on the driving}. A constraint on the driving is given in Eq.\eqref{driving_cond} and reads 
\begin{equation}
\frac{\max\limits_{n\neq m; t\in[0,t^*]} |\langle n_t|\frac{d}{dt}|m_t\rangle |}{\min\limits_{n\neq m; t\in[0,t^*]}|E_n(t)-E_m(t)|}= \frac{\Omega\omega}{4\omega_0^2}\ll a^{-1}  .
\end{equation}
Notice that for $a\sim 10^{-3}$ this implies a upper bound for the adiabatic parameter correspondent to $10^{3}$, which means that we can safely explore ultrastrong driving regimes without breaking this requirement.

\end{enumerate}

\begin{figure*}[ht!]
\centering
\includegraphics[width=0.9\textwidth]{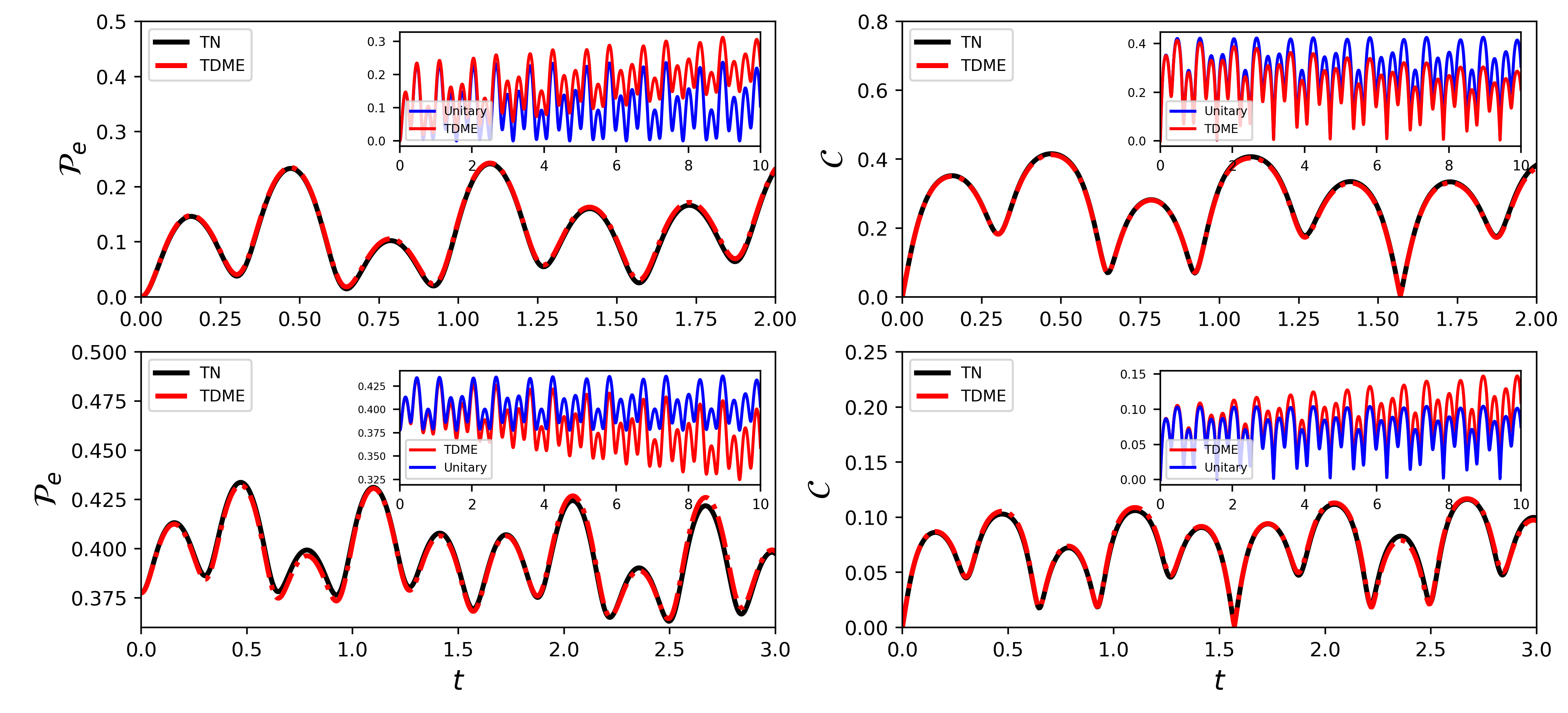}
\captionof{figure}{Population of the excited state $\mathcal{P}_e(t)$ (left column) and coherence $\mathcal{C}(t)$ (right column) in the strong driving regime $\lambda=\lambda_{\Omega}\lambda_{\omega}=10$ ($\lambda_{\Omega}=1$, $\lambda_{\omega}=10$). We plot different solutions, corresponding to tensor network simulations (TN), our time-dependent master equation (TDME) and the purely unitary evolution, i.e. the evolution of the system if $g=0$. The parameters $a=5\times 10^{-3}$ and $w_c=2$, which appear in the bath spectral density $J(w)=a we^{-w/w_c}$, are fixed. The system parameters are $\omega_0=1, \Omega=1, \omega=10$. Both system and environment are initially prepared in a thermal state at temperatures, respectively, $T_S=10^{-1}, T_B=4$ for the upper plots, $T_S=4, T_B=10^{-1}$ for the lower plots.
The oscillating behaviour is due to the large frequency difference $\omega\gg \omega_0$, which allows to easily achieve strong driving regimes even for $\lambda_{\Omega}\sim 1$. Close to resonance $\lambda_{\omega}\sim 1$, significant departures from adiabaticity ($\lambda\gg 1$) are accessible only for $\lambda_{\Omega}\gg 1$. However, this typically implies a large increase of the gap with respect to the dissipation rates, at the disadvantage of a dominant contribution of the unitary part of the dynamics.
The numerically exact tensor network simulations are given by a second order Time-Evolving Block Decimation (TEBD) for the system and environment density matrix, with $N=100$ harmonic oscillators. As explained more in detail in the Appendix C, the tensor network simulations can be exploited up to a certain $t_{max}$, mostly dependent on the finite number of oscillators and the temperature $T_B$; hence, the upper insets display the evolution in a wider time window. The bond dimension used in these simulations is $\chi=40$, whereas the maximum dimension of the oscillator Hilbert space in the environment is $d_{max}\sim 20$ (high temperature case).
When the environment temperature is larger than the system one, $T_B>T_S$, we observe wider excitations in the population $\mathcal{P}_e(t)$, respect to the purely unitary dynamics, whereas the coherence generated is lower. On the other hand, in the opposite case as $T_B<T_S$, the dynamics predicted by the TDME shows a higher amount of coherence generated, in agreement with the tensor network simulations,
while the excitations in the population are reduced, as expected for a system subject to cooling.} 
\label{fig: periodic_1}
\end{figure*}

As anticipated, the dynamics predicted by the master equation will be benchmarked against numerically exact tensor network simulations (\cite{TN_1,TN_2,TN_3}, while we refer those readers unfamiliar with this methodology to Appendix D for a brief introduction to the topic) and the so called adiabatic master equation \cite{Zanardi}
\begin{equation}
\begin{aligned}
\label{ME_adiab}
\frac{d}{dt}\rho_S=&-\imath[H_S(t)+H_{LS}^{ad}(t),\rho_S]\\
&+\sum\limits_{p=0,\pm}\gamma_p(t)\mathcal{D}[L_p^{ad}(t)]\rho_S,
\end{aligned}
\end{equation}
where the dissipation rates are exactly the same in Eq.\eqref{eq_rates}, while the jump operators and the Lamb shift Hamiltonian read 
\begin{equation}
\label{jump_ad}
\begin{aligned}
L_0^{ad}(t)=&\sin\varphi\Big(\sigma_z \cos\varphi+\sigma_x\sin\varphi\Big), \\
L_+^{ad}(t)=&\cos\varphi\Big( -\frac{1}{2}\sigma_z\sin\varphi\\&+\sigma_+\cos^2\frac{\varphi}{2}-\sigma_-\sin^2\frac{\varphi}{2} \Big), \\
L_-^{ad}(t)=&\cos\varphi\Big( -\frac{1}{2}\sigma_z\sin\varphi\\&+\sigma_-\cos^2\frac{\varphi}{2} -\sigma_+\sin^2\frac{\varphi}{2} \Big), \\
H_{LS}^{ad}(t)=& -\frac{1}{2}S(t)\cos^2(\varphi) H_S(t). \\
\end{aligned}
\end{equation}
The adiabatic condition in Eq.\eqref{adiabatic_cond} would now be expressed in the form:
\begin{equation}
\frac{\max\limits_{n\neq m; t\in[0,t^*]} |\langle n_t|\frac{d}{dt}|m_t\rangle |}{\min\limits_{n\neq m; t\in[0,t^*]}|E_n(t)-E_m(t)|}=  \frac{\Omega\omega}{4\omega_0^2}\ll 1  .
\end{equation}

For our analysis, we introduce the useful adimensional ratios $\lambda_{\Omega}\equiv \Omega/\omega_0$, $\lambda_{\omega}\equiv \omega/\omega_0$, such that $\lambda=\lambda_{\Omega}\lambda_{\omega}$ is the adiabatic parameter. 

In the following, we characterise some specific properties of the evolution under the master Eq.\eqref{generic_TDME}, that we will refer to TDME for short, benchmark it against tensor network simulations, the adiabatic master equation in Eq.\eqref{ME_adiab} (ADME) and the purely unitary dynamics as given by the Von Neumann equation. The differential equations which appear throughout this work are solved by means of a fourth-order Runge-Kutta method.

\paragraph{Strong driving regime: $\lambda\gg 1$.}

We begin by exploring the strong driving and dispersive regime, as given by $\lambda_{\Omega},\lambda_{\omega}\gg 1$.
In order to do so, we compute the population of the excited state $\mathcal{P}_{e}(t)=\langle e_t|\rho_S(t)|e_t\rangle$ and the coherence $\mathcal{C}(t)=|\langle e_t|\rho_S(t)|g_t\rangle|$, in the instantaneous eigenbasis $\{|n_t\rangle \}_{n=e,g}$ of $H_S(t)$, where the subscript e(g) indicates the excited (ground) state. 

Some results are shown in Fig.\eqref{fig: periodic_1}.  For $T_B>T_S$, the system is characterised by larger excitations in the population $\mathcal{P}_e$, with respect to the closed case (unitary dynamics).
On the other hand, for $T_S>T_B$ the TDME exhibits more generation of coherence than the purely unitary dynamics, as already observed in \cite{Dann}. 
This feature is uniquely due to the dissipator and does not stem from the Lamb shift contribution, which is actually negligible in our specific case (see discussion below).

\begin{figure*}[ht!]
\centering
\includegraphics[width=0.9\textwidth]{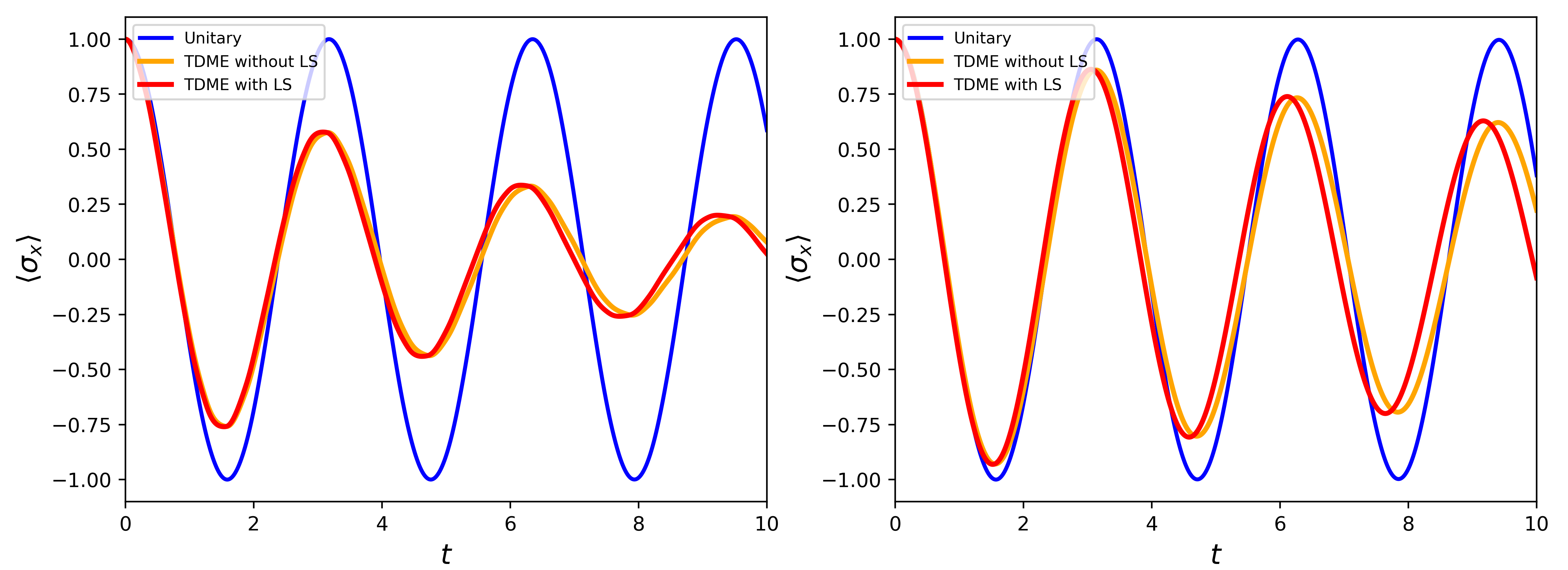}
\captionof{figure}{Transverse magnetization $\langle\sigma_x(t)\rangle$. Analysis of the Lamb shift contribution. The system is initially prepared in a coherent superposition $|\psi\rangle=\frac{1}{\sqrt{2}}(|g_0\rangle + |e_0\rangle)$, such that $\rho_{12}(0)\neq 0$, as explained in the main text. The spectral density parameters and temperature are fixed and given by $a=5\times 10^{-3}$, $w_c=2$, $T_B=4$, whereas in the plot on the left we have $\omega_0=1,\Omega=1,\omega=10$ ($\lambda=10$) and on the right is $\omega_0=1,\Omega=10^{-1},\omega=10^{-1}$ ($\lambda=10^{-2}$).
The plots show the expectation values obtained via TDME in Eq.\eqref{generic_TDME}, with and without Lamb shift term. The main effect of the Lamb shift is to produce a shift in the curves, as already pointed out in \cite{Rivas_2}; however, differently from this previous work, in the presence of driving this effect is not uniform throughout the whole evolution and it can change over time, as we notice in the right plot. }
\label{fig: Periodic_2}
\end{figure*}

\paragraph{Effect of Lamb shift contribution.}

In the case of time-independent Markovian master equations, the Lamb shift contribution commutes with the system Hamiltonian and typically produces a small shift in the energy levels; therefore, its contribution is commonly neglected, although some works provided cases where it gives sizeble effects (as \cite{Rivas_2} for the time-independent case, or \cite{Whitney} in the presence of driving). 
 
In order to analyse the effect of the Lamb shift contribution it is convenient to revert to 
the interaction picture. In this case, the master Eq.\eqref{generic_TDME} assumes the simple form
\begin{equation}
\begin{aligned}
\frac{d}{dt}\tilde{\rho}_S=&-\imath[\tilde{H}_{LS}(t),\tilde{\rho}_S] \\
&+\sum\limits_{p=0,\pm}\gamma_p(t)\mathcal{D}[A_p(t)]\tilde{\rho}_S,
\end{aligned}
\end{equation}
with
\begin{equation}
\begin{aligned}
A_0&=\sigma_z\sin\varphi(t), \\
A_+&=\sigma_+\cos\varphi(t), \\
A_-&=\sigma_-\cos\varphi(t), \\
\tilde{H}_{LS}&=-\frac{1}{2}\sigma_z S(t)\cos^2\varphi(t) .
\end{aligned}
\end{equation}
The density matrix can be parametrised in terms of the expectation values of $\sigma_z,\sigma_{\pm}$ as 
\begin{equation}
\tilde{\rho}_S(t)=\frac{1}{2}\Big(1+\sum\limits_{k=z,\pm}\tilde{c}_k(t)\sigma_k \Big),
\end{equation}
where $\tilde{c}_k(t)=tr[\sigma_k\tilde{\rho}_S(t)]$, $k=z,\pm$, are analogous to the components of the Bloch vector of the system. Using this representation, 
the equations of motion for $\tilde{c}_k(t)$ can be easily integrated, leading to the general solution
\begin{equation}
\begin{aligned}
\tilde{c}_z(t)=&c_z(0)e^{-\zeta(t)} +\int\limits_0^t ds \Big(\gamma_+(s)
-\gamma_-(s)\Big)\\ &\times \cos^2\varphi(s)e^{-(\zeta(t)-\zeta(s))}, \\
\tilde{c}_+(t)=&c_+(0)\exp\Big\{ -\imath\int\limits_0^t ds S(s)\cos^2\varphi(s)\Big\}\\
&\times \exp\Big\{ -\frac{1}{2}\zeta(t)-2\gamma_0\int\limits_0^t ds \sin^2\varphi(s)\Big\} , \\
\end{aligned}
\end{equation}
with $\zeta(t)=\int\limits_0^t ds (\gamma_+(s)+\gamma_-(s))\cos^2\varphi(s)$. 

We notice that the contribution given by the Lamb shift enters only in the off-diagonal terms, therefore if the initial state has zero coherence in the eigenbasis of $H_S(0)$, i.e. $\rho_{12}(0)\equiv \frac{1}{2}c_+(0)= 0$, the Lamb shift does not play any role, regardless of the choice of the temperature $T_B$ and the driving. This is the case of thermal states studied so far. As a consequence, with the purpose of highlighting the impact of the Lamb shift on the dynamics, we consider an initial preparation for the two-level system in a coherent superposition
given by $|\psi(0)\rangle=\frac{1}{\sqrt{2}}(|g_0\rangle + |e_0\rangle)$ and compare the dynamics generated by the master Eq.\eqref{generic_TDME} with and without Lamb shift, as shown in Fig.\eqref{fig: Periodic_2}.
For both cases of strong and adiabatic driving (see below for a detailed discussion), given respectively by $\lambda=10$, $\lambda=10^{-2}$, we observe that the Lamb shift produces a shift in the transverse magnetization $\langle\sigma_x(t)\rangle$, which is not uniform throughout the time evolution and has a different magnitude depending on the strength of the driving.

\paragraph{Adiabatic regime: $\lambda\ll 1$.}

When the adiabatic condition \eqref{adiabatic_cond} is satisfied, we can formally replace $U_S(t)\rightarrow U_S^{ad}(t)$ in Eq.\eqref{time-dep-Markov} and the TDME reproduces the ADME by construction. 

In order to state this precisely, let us study the explicit form of the time-evolution operator when $\lambda\ll 1$. We consider initially the case $\lambda_{\Omega}\ll 1$, which we simply refer to as weak driving, while $\lambda_{\omega}$ is arbitrary for the moment. From Eq.\eqref{system_param} we perform the substitution $x=\omega_0 t\equiv t/t_s$ obtaining
\begin{equation}
\begin{aligned}
\imath\frac{d}{dx}\alpha&=\alpha-\lambda_{\Omega} \sin\Big(\lambda_{\omega}x\Big)\beta^*, \\
\imath\frac{d}{dx}\beta^*&=-\beta^*-\lambda_{\Omega}\sin\Big(\lambda_{\omega}x\Big)\alpha. \\
\end{aligned}
\end{equation}
We can find approximate solutions expanding  in powers of $\epsilon$ both $\alpha(x)=\sum\limits_{n=0}^{\infty} \epsilon^n \alpha_n(x)$ and $\beta(x)=\sum\limits_{n=0}^{\infty} \epsilon^n \beta_n(x)$, where $\epsilon$ is a small adimensional parameter, such as $\lambda_{\Omega}$ in this case.
Up to second order corrections we find
\begin{equation}
\label{param_periodic_weak}
\begin{aligned}
\alpha(x)=&e^{-\imath x} +O(\lambda_{\Omega}^2), \\
\beta(x) =& \frac{\imath}{2}\lambda_{\Omega}e^{-\imath x}\Big(\frac{e^{-\imath(\lambda_{\omega}-2)x }-1}{\lambda_{\omega}-2}  \\
&+ \frac{e^{\imath(\lambda_{\omega}+2)x }-1}{\lambda_{\omega}+2} \Big) +O(\lambda_{\Omega}^2).\\
\end{aligned}
\end{equation}
Furthermore, if we consider the additional requirement $\lambda_{\omega}\ll 1$ in Eq.\eqref{param_periodic_weak} a first order expansion gives
\begin{equation}
\label{param_period_weak}
\begin{aligned}
\alpha(x)=&e^{-\imath x} +O(\lambda_{\Omega}^2), \\
\beta(x) =&\frac{\imath}{4}\lambda_{\Omega}\lambda_{\omega}(e^{-\imath x}-e^{\imath x}\cos(\lambda_{\omega} x)) \\
&-\frac{1}{2}\lambda_{\Omega} \sin(\lambda_{\omega} x)e^{\imath x}+O(\lambda_{\Omega}^2; \lambda_{\omega}^2).\\
\end{aligned}
\end{equation}

The conditions $\lambda_{\Omega},\lambda_{\omega}\ll 1$ allow us to achieve adiabaticity.
Notice that, for instance, one can also consider to drastically reduce $\lambda_{\Omega}$ keeping $\lambda_{\omega}\sim 1$ and still obtain $\lambda\ll 1$, although this typically requires considerably small values for the parameters in the Hamiltonian.

\begin{figure*}[ht!]
\centering
\includegraphics[width=0.9\textwidth]{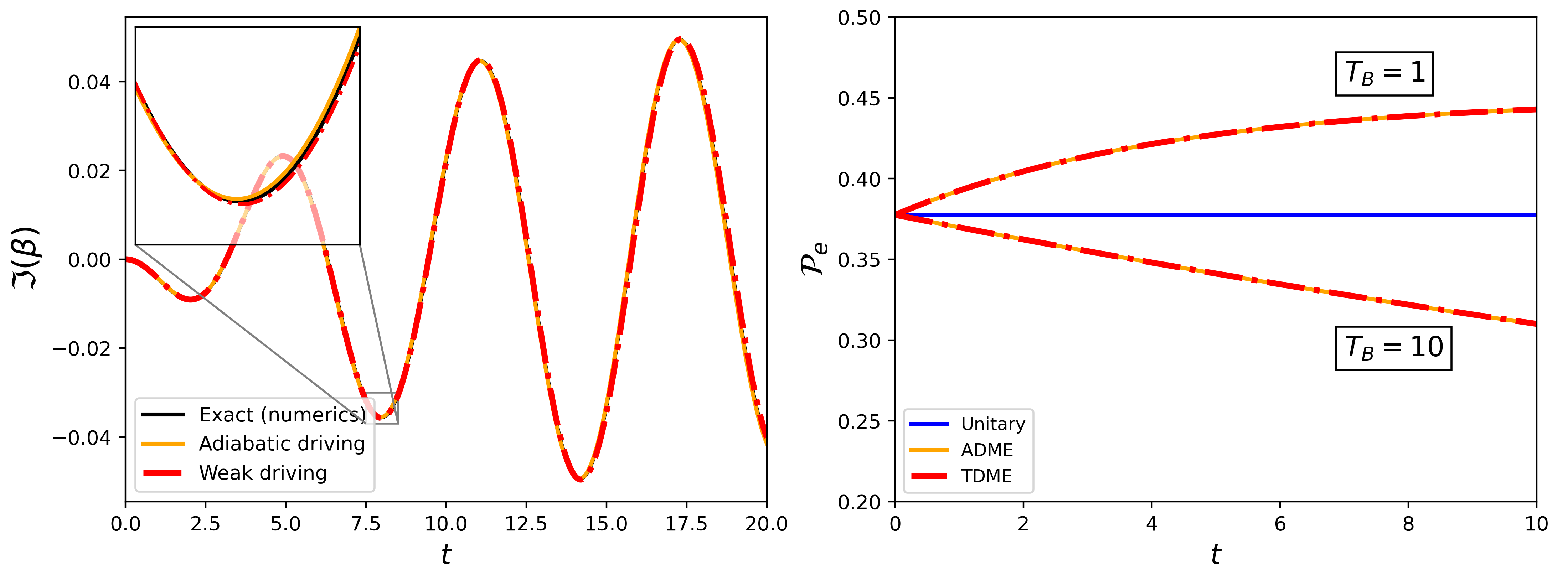}
\captionof{figure}{Left panel: imaginary part of $\beta(t)$. The spectral density parameters are $a=5\times 10^{-3}$, $w_c=2$. The system parameters are $\omega_0=1,\Omega=10^{-1},\omega=10^{-1}$.
The three curves, correspondent to the exact numerical solution of Eq.\eqref{system_param}, the approximated analytic solution in Eq.\eqref{param_period_weak} (red dashed line) and \eqref{param_period_ad} (orange line), are overlapped for $\lambda_{\Omega}=\lambda_{\omega}=10^{-1}$ and therefore the closed dynamics is well approximated by the adiabatic time-evolution operator. 
Right panel: Population of the excited state $\mathcal{P}_e$. The initial state is a  Gibbs state at temperature $T_S=4$ with respect to $H_S(0)$. For $\lambda=10^{-2}$, the TDME reduces to the ADME, as shown for both the regimes $T_B>T_S$, $T_B<T_S$.} 
\label{fig: periodic_3}
\end{figure*}

The adiabatic time-evolution operator in Eq.\eqref{ad_op_gen} reads here
\begin{equation}
\label{time_ev_ad}
\left(
\begin{array}{cc}
\alpha^{ad} & \beta^{ad} \\
-\beta^{ad *} & \alpha^{ad *} \\
\end{array}
\right) ,
\end{equation}
where
\begin{equation}
\begin{aligned}
\alpha^{ad}&=\exp\Big\{-\imath\int\limits_0^t ds \sqrt{\omega_0^2+\Omega^2\sin^2(\omega s)} \Big\} \cos\frac{\varphi}{2} ,\\
\beta^{ad}&=-\exp\Big\{ \imath\int\limits_0^t ds \sqrt{\omega_0^2+\Omega^2\sin^2(\omega s)} \Big\} \sin\frac{\varphi}{2} .\\
\end{aligned}
\end{equation}

We expand these expressions in powers of $\lambda_{\Omega},\lambda_{\omega}\ll 1$, arriving at
\begin{equation}
\label{param_period_ad}
\begin{aligned}
\alpha^{ad}&\stackrel{\rm x=\omega_0 t }{=}e^{-\imath x}+O(\lambda_{\Omega}^2) ,\\
\beta^{ad}&\stackrel{\rm x=\omega_0 t }{=}-\frac{1}{2}\lambda_{\Omega}\sin(\lambda_{\omega}x) e^{\imath x}+ O(\lambda_{\Omega}^2; \lambda_{\omega}^2).\\
\end{aligned}
\end{equation}
From a comparison with Eq.\eqref{param_period_weak}, we see that $\alpha(x)-\alpha^{ad}(x)=O(\lambda_{\Omega}^2)$ and $\beta(x)-\beta^{ad}(x)=O(\lambda_{\Omega}\lambda_{\omega})$,  hence providing the condition $O(\lambda_{\Omega})=O(\lambda_{\omega})$ to be satisfied the closed dynamics approaches adiabaticity. 
A specific example is illustrated in Fig.\eqref{fig: periodic_3}, where already for $\lambda_{\Omega}=\lambda_{\omega}=10^{-1}$ we find a good overlapping between the solutions given by the TDME and ADME.

\section{Second example: interacting qubits}

As a second example, we study the dynamics of two interacting qubits in the presence of both dissipation and driving. 
The analysis of an open many-body system, whose components are locally coupled to some environment, raises the question of whether a master equation description in terms of local jump operators is possible or not.
This so-called 'local vs global' description has attracted great interest, as witnessed by the vast literature on the topic (as for example \cite{Rivas_2,Alicki_0,Levy,Trushechkin_1,Gonzales,Hofer,Mitchison,Dechiara,Cattaneo,Benatti,
Farina_2,Potts,Scali,Konopik,Tupkary}), although the case with driven dissipative systems has received significantly less attention.

In the model we analyse, we assume that only one qubit directly interacts with the environment, whereas the other one is driven, as described by the total Hamiltonian
\begin{equation}
H=H_S(t) + \sum\limits_j w_j b_j^{\dagger} b_j + \sigma_x^{(2)}\otimes\sum\limits_j g_j(b_j+b_j^{\dagger}),
\end{equation}
with system Hamiltonian
\begin{equation}
H_S(t)= \frac{\omega_1(t)}{2}\sigma_z^{(1)}+\frac{\omega_2}{2}\sigma_z^{(2)}+\lambda(\sigma_+^{(1)}\sigma_-^{(2)}+h.c).
\end{equation}
We assume that the driving consists of a periodic modulation of the first qubit frequency,
\begin{equation}
\omega_1(t)=\omega_2+\delta\sin(\eta t),
\end{equation}
with $\omega_2,\delta,\eta>0$ and we always consider $\omega_2\ge \delta$, in order to keep $\omega_1(t)\ge 0, \forall t$. Moreover, it is convenient to introduce the parameters $\omega_{\pm}(t)=\frac{1}{2}(\omega_1(t)\pm\omega_2)$ and $\tan(2\theta(t))=\lambda/\omega_-(t)$. Differently from the previous case, it will be sufficient for our purposes to consider the environment at zero temperature $T_B=0$.

As shown in Appendix B, the master equation in the Schrödinger picture can be expressed in the form
\begin{equation}
\begin{aligned}
\label{me_twoqubits}
\frac{d}{dt}\rho_S=&-\imath[H_S(t)+H_{LS}(t),\rho_S] \\
&+\sum\limits_{p=a,b}\gamma(\Omega_{p}(t))\mathcal{D}[L_p(t)]\rho_S,
\end{aligned}
\end{equation}
where 
\begin{equation}
\begin{aligned}
L_a(t)&=\sin\theta\Big(|\psi_0(t)\rangle \langle\psi_2(t)|  -|\psi_1(t)\rangle \langle\psi_3(t)|  \Big), \\
L_{b}(t)&=\cos\theta\Big(|\psi_0(t)\rangle \langle\psi_1(t)|  +|\psi_2(t)\rangle \langle\psi_3(t)|  \Big), \\
\end{aligned}
\end{equation}
are the jump operators, with $|\psi_n\rangle$ solution of the Schrödinger equation $\imath\frac{d}{dt}|\psi_n\rangle=H_S(t)|\psi_n\rangle$ with initial condition the eigenvector $|E_n(0)\rangle$ of $H_S(0)$.
The instantaneous Bohr frequencies are given by
\begin{equation}
\begin{aligned}
\Omega_{a}\equiv\omega_+(t) + \sqrt{\lambda^2+ \omega_-^2(t)},  \\
\Omega_{b}\equiv\omega_+(t) - \sqrt{\lambda^2+ \omega_-^2(t)}, \\
\end{aligned}
\end{equation}
\begin{equation}
\begin{aligned}
H_{LS}(t)=\sum\limits_{p=a,b}\Big( & S(\Omega_p(t))L_{p}^{\dagger}(t)L_{p}(t)\\
+ & S(-\Omega_p(t))L_{p}(t)L_{p}^{\dagger}(t)\Big),
\end{aligned}
\end{equation}
is the Lamb shift contribution, whereas
\begin{equation}
\label{twoqubit_corr}
\begin{aligned}
\gamma(\Omega)&=\begin{cases}
2\pi J(\Omega)  \text{\space\space if $\Omega>0$}, \\
0  \text{\space\space\space\space\space\space\space\space\space\space otherwise},\\
\end{cases} \\
S(\Omega)&= \mathcal{P}\int\limits_0^{\infty} dw \frac{J(w)}{\Omega-w},
\end{aligned}
\end{equation}
assuming as before an ohmic spectral density $J(w)=a w e^{-w/w_c}$.
We refer the reader to Appendix B for the detailed derivation.

Besides the condition on the environment correlation function and weak coupling, already studied in the previous example, the validity of the master equation Eq.\eqref{me_twoqubits}  is assured by a set of sufficient conditions, namely:

\begin{enumerate}
\item \textit{Full-secular approximation}.
The explicit expression of this condition depends on the values of the parameters, according to the prescription given in Eq.\eqref{general_secular}. For the sake of notation, we define the intervals
\begin{equation}
\begin{aligned}
I_1&=\Big(0, \sqrt{(\omega_2-\delta)\omega_2}\Big), \\
I_2&=\Big[\sqrt{(\omega_2-\delta)\omega_2 },\sqrt{(\omega_2+\delta)\omega_2}\Big], \\
I_3& = \Big(\sqrt{(\omega_2+\delta)\omega_2},\infty\Big),
\end{aligned}
\end{equation}
for the coupling $\lambda$.

\begin{enumerate}
\item If $\lambda\in I_1$, the condition reads
\begin{equation}
\begin{aligned}
\label{sec_cond_1}
\min\{2\omega_2-\delta-\sqrt{4\lambda^2+\delta^2}, 2\lambda \} \gg a w_c.\\
\end{aligned}
\end{equation}

\item If $\lambda\in I_2 $, two separate conditions are needed:

\begin{equation}
\label{twoqubits_sec_driving}
\begin{aligned}
\min\{2\omega_2-\delta,2\lambda \}&\gg a w_c, \\
\eta\delta &\gg  4 a w_c^2.
\end{aligned}
\end{equation}

Notice that, differently from the other ranges of parameter, here the secular approximation requires explicitly a lower bound on the 'speed' of the driving $\eta$.

\item If $\lambda\in I_3$, the condition reads
\begin{equation}
\begin{aligned}
\min\{2\sqrt{\lambda^2+(\delta/2)^2}  -2\omega_2-\delta, \\
2\omega_2-\delta\} \gg a w_c.
\end{aligned}
\end{equation}

\end{enumerate}

\begin{figure*}[ht!]
\centering
\includegraphics[width=0.9\textwidth]{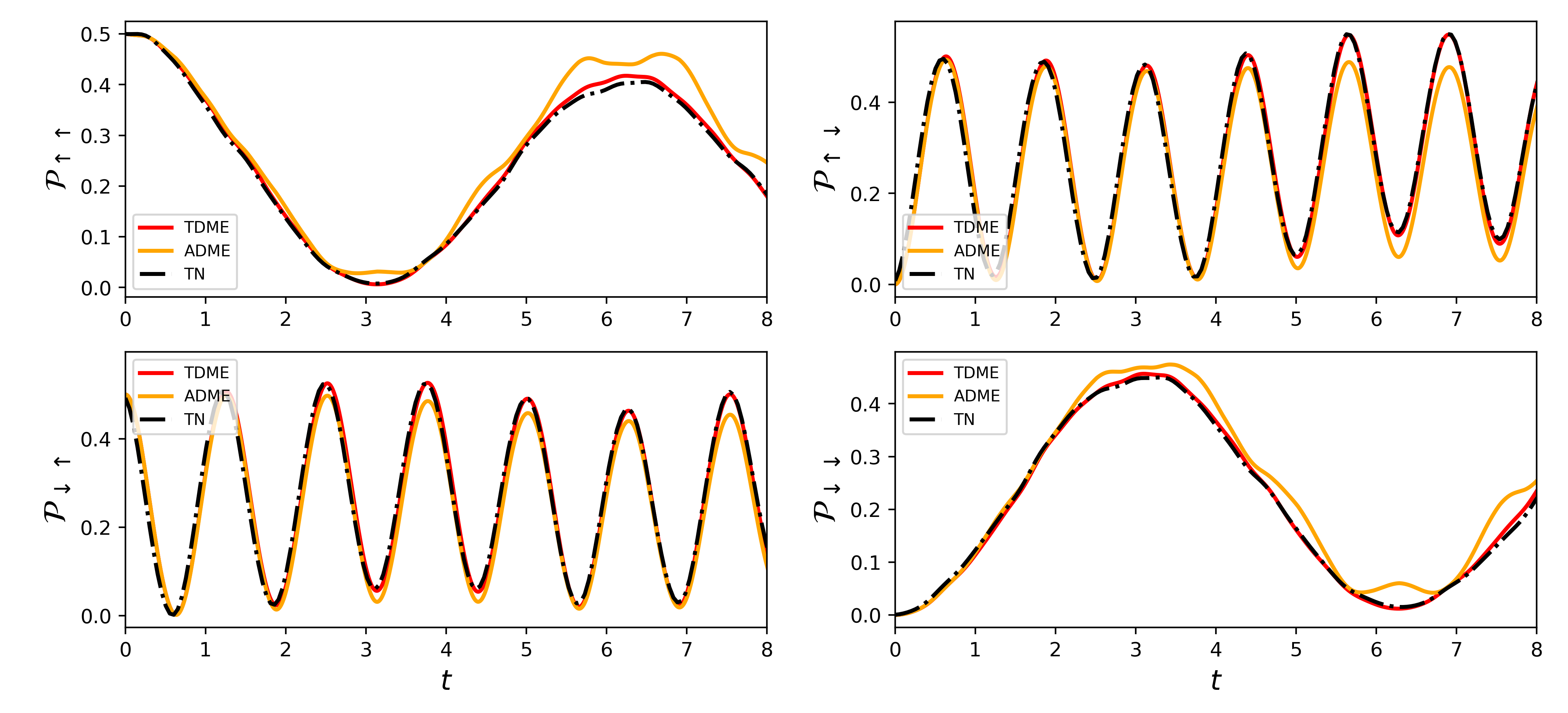}
\captionof{figure}{Populations $\mathcal{P}_{n,m}(t)=\langle n|_1\otimes \langle m|_2\rho_S(t)|n\rangle_1 \otimes |m\rangle_2 $ relative to the eigenstates $\{|n\rangle\}_{n=\uparrow,\downarrow}$ of $\sigma_x$.  We prepare the system in the coherent superposition $|\psi(0)\rangle=|1\rangle_1\otimes \frac{1}{\sqrt{2}}(|0\rangle_2 + |1\rangle_2  )$, where $|0\rangle,|1\rangle$ are the eigenstates of $\sigma_z$. The different plots show the dynamics given by the Eq.\eqref{me_twoqubits} (TDME), the adiabatic master Eq.\eqref{me_twoqubits_ad} (ADME) and the tensor network simulation (TN). The system parameters are, in arbitrary units, $\omega_2=2,\delta=2\times 10^{-1},\eta=10,\lambda=3$, while for the environment $a=5\times 10^{-3},w_c=6$. Tensor networks simulations are performed considering $N=100$ harmonic oscillators, with  $\chi=60$ the bond dimension and $d_{max}=3$ the highest local dimension of the environment. }
\label{fig: twoqubits_1}
\end{figure*}

\item \textit{Condition on the driving}. As for the secular approximation, we obtain a different condition depending on the value of the coupling $\lambda$. In general, we can summarise it as
\begin{equation}
\label{twoqubits_drivingcon}
\frac{\delta\eta}{4\lambda}\le a^{-1} \min\mathcal{A} , 
\end{equation}
where the set $\mathcal{A}$ reads
\begin{equation}
\begin{aligned}
\Big\{2\lambda, \omega_2-\frac{\delta}{2}-\sqrt{\lambda^2+\frac{\delta^2}{4}} \Big\} \text{\space if $\lambda\in I_1$ }, \\
\{2\omega_2-\delta,2\lambda \} \text{\space if $\lambda\in I_2$ }, \\
\Big\{2\omega_2-\delta,\sqrt{\lambda^2+\frac{\delta^2}{4}}-\omega_2-\frac{\delta}{2} \Big\} \text{\space if $\lambda\in I_3$ }. \\
\end{aligned}
\end{equation}

\end{enumerate}

In the following analysis, we will focus on local observables of the single qubits.
This will be done by comparing the master Eq.\eqref{me_twoqubits}, that we refer to TDME again, benchmarking it against tensor network simulations, with the same technique already illustrated in the previous section and Appendix D, and the purely unitary dynamics as given by the Von Neumann equation. In addition, we employ also a 'local' version of master equation (Eq.\eqref{local_me_twoqubits} below) in the limit of weakly interacting qubits and the adiabatic master equation \cite{Zanardi}, which can be expressed for our model in the form
\begin{equation}
\begin{aligned}
\label{me_twoqubits_ad}
\frac{d}{dt}\rho_S=&-\imath[H_S(t)+H_{LS}^{ad}(t),\rho_S] \\
&+\sum\limits_{p=a,b}\gamma(\Omega_{p}(t))\mathcal{D}[L_p^{ad}(t)]\rho_S,
\end{aligned}
\end{equation}
where
\begin{equation}
\begin{aligned}
L_a^{ad}(t)=&\sin\theta(t)\Big( |E_0(t)\rangle\langle E_2(t)|\\
&-|E_1(t)\rangle\langle E_3(t)|\Big), \\
L_b^{ad}(t)=&\cos\theta(t)\Big( |E_0(t)\rangle\langle E_1(t)|\\
&+|E_2(t)\rangle\langle E_3(t)|\Big), \\
H_{LS}^{ad}(t)=&\sum\limits_{p=a,b}\Big( S(\Omega_p(t))L_{p}^{ad \dagger}(t)L_{p}^{ad}(t)  \\
&+ S(-\Omega_p(t))L_{p}^{ad}(t)L_{p}^{ad \dagger}(t)\Big).
\end{aligned}
\end{equation}

\begin{figure*}[ht!]
\centering
\includegraphics[width=0.9\textwidth]{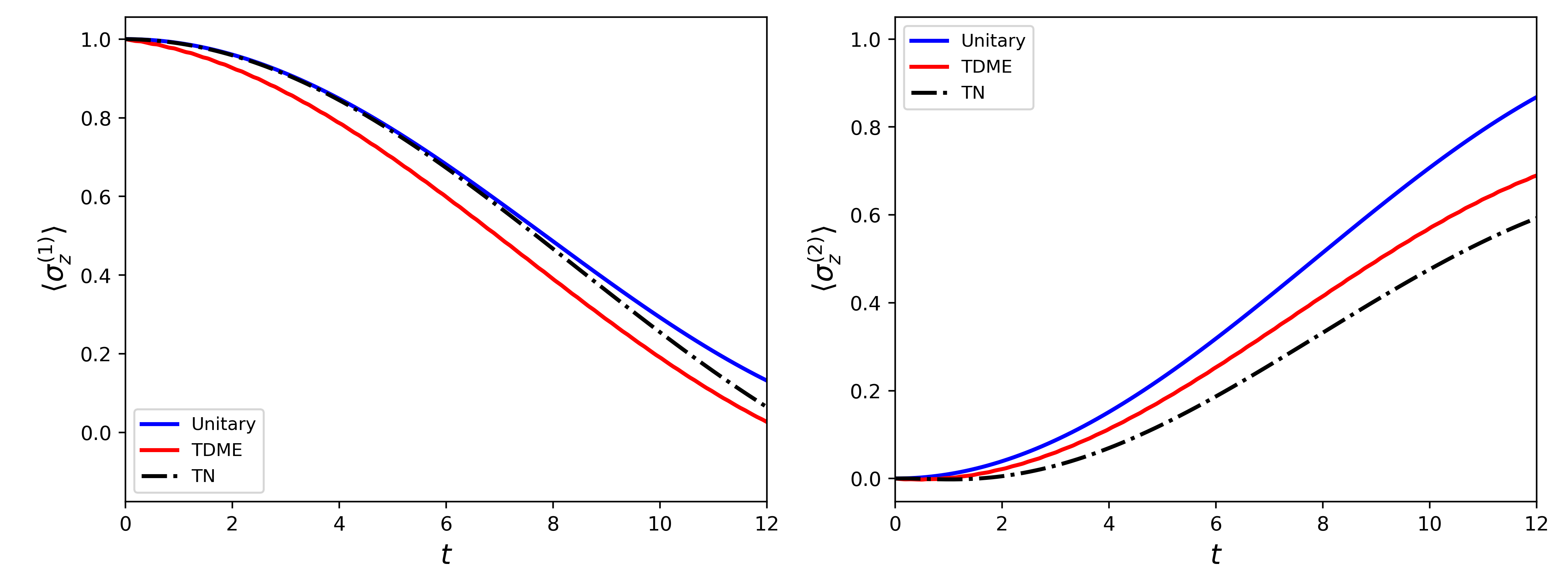}
\captionof{figure}{Longitudinal magnetisation of qubit 1 and 2. The system is initially prepared in the state $|\psi\rangle=|1\rangle_1\otimes\frac{1}{\sqrt{2}}\Big(|0\rangle_2 + |1\rangle_2\Big)$. The results shown correspond to Eq.\eqref{me_twoqubits} (TDME), the purely unitary dynamics and the tensor network simulation (TN). The system parameters are, in arbitrary units, $\omega_2=1,\delta=1,\eta=20,\lambda=10^{-1}$, while for the environment $a=5\times 10^{-3},w_c=4$. The parameters related to the tensor networks are the same as in Fig.\eqref{fig: twoqubits_1}. The condition on the driving in Eq.\eqref{twoqubits_drivingcon} is violated for this set of parameters and the TDME shows  a deviation from the tensor network simulation. } 
\label{fig: twoqubits_2}
\end{figure*}

\paragraph{Strong driving regime.}

We start by analysing the dynamics described by our master equation Eq.\eqref{me_twoqubits} in the strong driving regime, as shown in Fig.\eqref{fig: twoqubits_1}.  The system is prepared in
the coherent superposition $|\psi(0)\rangle=|1\rangle_1\otimes \frac{1}{\sqrt{2}}(|0\rangle_2 + |1\rangle_2  )$, where $|0\rangle,|1\rangle$ are the eigenstates of $\sigma_z$ and we compute the populations $\mathcal{P}_{n,m}(t)=\langle n|_1\otimes \langle m|_2\rho_S(t)|n\rangle_1 \otimes |m\rangle_2 $ relative to the eigenstates $\{|n\rangle\}_{n=\uparrow,\downarrow}$ of $\sigma_x$. In this regime, we clearly see in Fig.\eqref{fig: twoqubits_1} that the adiabatic master equation does not correctly reproduce the dynamics, while Eq.\eqref{me_twoqubits} shows a good agreement with the tensor network simulations.

Furthermore, in order to test the validity of our results, we consider in  Fig.\eqref{fig: twoqubits_2} a case of explicit violation of the condition on the driving in Eq.\eqref{twoqubits_drivingcon}, which leads to a expected deviation from the TN simulations.

\paragraph{Weakly interacting qubits.}

\begin{figure*}[ht!]
\centering
\includegraphics[width=0.9\textwidth]{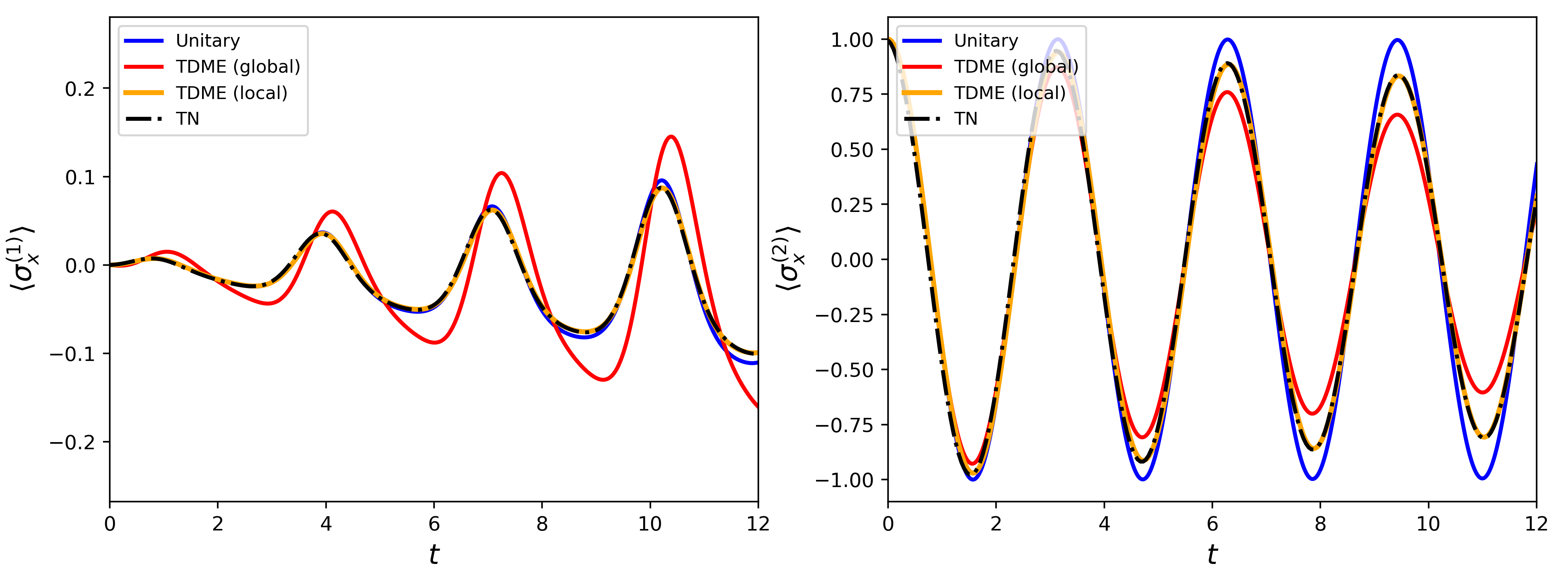}
\captionof{figure}{Transverse local magnetisation of qubit 1 and 2. The system is initially prepared in the product state $|\psi\rangle=|1\rangle_1\otimes\frac{1}{\sqrt{2}}\Big(|0\rangle_2 + |1\rangle_2\Big)$. The curves shown correspond to Eq.\eqref{me_twoqubits} (TDME global), Eq.\eqref{local_me_twoqubits} (TDME local), the purely unitary dynamics and the tensor network simulation (TN). The system parameters are, in arbitrary units, $\omega_2=2,\delta=1,\eta=2,\lambda=10^{-2}$, while for the environment $a=5\times 10^{-3},w_c=4$. The parameters considered in the tensor network simulations are the same as in Fig.\eqref{fig: twoqubits_1}. For weak qubit-qubit interaction,  qubit 2 is mainly subject to dissipation, as correctly predicted by the local master equation, while qubit 1 displays a dynamics very close the unitary. Notice that, as explained in the main text, small deviations from the unitary dynamics can still be observed for qubit 1 at later times, due to presence of a non-zero qubit-qubit interaction in the unitary part of the master equation Eq.\eqref{local_me_twoqubits}.
Due to the breakdown of the full-secular approximation in Eq.\eqref{sec_cond_1}, the global master equation is not able to correctly reproduce the dynamics.  } 
\label{fig: twoqubits_3}
\end{figure*}

In a many-body system, correlations and excitations locally created by the coupling with an environment can propagate very fast to distant regions, especially in the presence of strong interactions among the system components.
Therefore, a master equation description frequently requires  non-local jump operators, namely operators with non-zero expectation values on many-particle states. However, a different situation can manifest when the internal interactions are sufficiently weak, with respect to the typical system-environment coupling strength.

In this last scenario, one can expect for our specific example that the open dynamics of the two qubits is partially decoupled, with the dissipative part acting only on the qubit directly in contact with the environment. 

However, when $\lambda$ becomes very small the energy spectrum manifests quasi-degeneracies and the full-secular approximation in Eq.\eqref{sec_cond_1} may break down, jeopardising the validity of Eq.\eqref{me_twoqubits}.

Nevertheless, providing the hypothesis that $\lambda=O(g)$, a local master equation can be directly derived from the Nakajima-Zwanzig equation, as shown in the Appendix B. The intuitive idea is that, for sufficiently weak interaction, the effects of the environment can be effectively described in terms of eigenstates of the local Hamiltonians $H_{1/2}=\frac{\omega_{1/2}}{2}\sigma_z^{(1/2)}$.
This master equation has the form
\begin{equation}
\begin{aligned}
\label{local_me_twoqubits}
\frac{d}{dt}\rho_S=&-\imath[H_S(t)+H_{LS},\rho_S]\\
&+\gamma(\omega_2)\Big(\sigma_-^{(2)}\rho_S\sigma_+^{(2)}-\frac{1}{2}\{ \sigma_+^{(2)}\sigma_-^{(2)},\rho_S\} \Big),
\end{aligned}
\end{equation}
where $H_{LS}=\frac{1}{2}(S(\omega_2)-S(-\omega_2))\sigma_z^{(2)}$. 
Whereas another full-secular approximation is needed also here, namely $2\omega_2\gg a C$, there is no constraint on the driven frequency $\omega_1(t)$.
As expected, the dissipation acts directly on qubit 2. However, given that the unitary part still contains the interaction term $\lambda(\sigma_+^{(1)}\sigma_-^{(2)}+h.c.)$, we underline that the dynamics is not fully decoupled and for large times dissipative effects can be observed also on qubit 1.

A numerical analysis which summarises the content of this subsection is shown 
in Fig.\eqref{fig: twoqubits_3}, where we compare  Eq.\eqref{me_twoqubits} (TDME global) with Eq.\eqref{local_me_twoqubits} (TDME local), the tensor network simulations and the unitary evolution, observing good agreement with the arguments outlined above.

\section{Conclusions}

We have derived a time-dependent Markovian master equation for a system subject to driving fields, starting from the Nakajima-Zwanzig equation in an interaction picture. The procedure
is based on the proof of convergence of the Nakajima-Zwanzig equation for the system density matrix to a Markovian evolution, under a weak-coupling limit $g\rightarrow 0$ between the system and the environment and makes use of a time rescaling $t=\tau/g^2$.

In order to perform this limit consistently, we have introduced a renormalisation of the system parameters $\{ \lambda \}$ which appear in the Hamiltonian $H_S(t, \lambda)$,  to reabsorb the dependence on $g$ in $H_S$ after the time rescaling, extending the canonical procedure outlined by Davies (\cite{Davies_1, Davies_2}).

We have developed a systematic and accurate derivation, incorporating two key ingredients. First, we utilise the method of stationary phase, a proven effective tool for generating accurate bounds for the approximations involved. Secondly, we employ a lemma which demonstrates that, when subject to time rescaling and weak coupling, only adiabatic contributions to the time-evolution operator $U_S(t)$ are relevant. This allows us to include non-adiabatic contributions naturally, as we utilise the entire time evolution operator to express the equation in the Schrödinger picture. Together, these ingredients enable us to achieve a highly accurate and rigorous derivation.

The validity of the derived master equation is characterised by a series of sufficient conditions, namely weak-coupling, integrability of the environment correlation functions, a full-secular approximation and a condition on the driving which is sufficiently loose to allow us to explore strong driving regimes.

We tested the range of validity of our master equation by studying the spin-boson model with a single periodically driven two-level system. Our findings suggest that our master equation provides an accurate and feasible description for strong periodic driving fields, alternative to the standard Floquet theory based approach \cite{Floquet}. We compared our master equation against numerically exact results given by tensor network simulations and considered also the adiabatic master equation of \cite{Zanardi} as comparison. Our master equation matches very well the tensor network simulations even for very strong drivings, as given by $\lambda\sim 10$, with $\lambda$ the adiabatic parameter.

In the case of strong driving, we observed that the environment plays a more prominent role in generating coherence at lower temperatures than purely unitary evolution, which has also been noted in previous works \cite{Dann}.
We have also studied in detail the effect of the Lamb shift contribution on the dynamics. In particular, we have proved that if the initial state is diagonal in the eigenbasis of $H_S(0)$, the Lamb shift does not contribute to the evolution. On the other hand, when some initial coherence is introduced, the Lamb shift has sizable impact on the dynamics, as observed in the transverse magnetization $\langle\sigma_x(t)\rangle$.
In the last part, we conducted a rigorous analysis of the adiabatic limit.

Finally, we moved from a single qubit to two interacting qubits, where only one of them is directly in contact with the environment, while the other one is subject to a periodic modulation of its bare gap. We presented the analysis of local observables of the two qubits.
In the strong driving regime, our master equation is proven to correctly reproduce the exact dynamics, moreover we put to the test the bounds given by the condition on the driving outlined in Eq.\eqref{driving_cond}, checking that a violation of them does imply that the master equation ceases to be valid. The last analysis covers in detail the case of weakly interacting qubits. It is shown that when the qubit-qubit interaction is of the same order of magnitude of the system-environment coupling, the full-secular approximation breaks down. In this regime, we provided a local master equation that can correctly describe the system dynamics. Ultimately, a global master equation with a partial secular approximation would be able to overcome the limitation that our approach suffers in this situation and would also represent a step forward towards a thorough analysis of the 'local vs global' debate in the presence of strong drivings.
This, together with a comprehensive investigation of both transient and steady state properties of the driven dissipative dynamics described by our derived master equation, will be the topic of future work.

\section{Acknowledgements}

Giovanni Di Meglio wishes to thank Nicola Lorenzoni, Dario Cilluffo and Raphael Weber for their precious suggestions and discussions. This work was supported by the ERC Synergy grant HyperQ (Grant no. 856432) and the BMBF project CoGeQ (Grant no. 13N16101).
The authors express their gratitude to the anonymous referees for their comments, which helped to substantially improve the manuscript.

\onecolumn
\appendix

\section{Derivation of the master equation for the driven two-level system}

In this Appendix, we derive the master equation Eq.\eqref{time-dep-Markov} for a driven two-level system described by the Hamiltonian 
\begin{equation}
H(t)=H_S(t) + \sum\limits_j w_j b^{\dagger}_j b_j +\sum\limits_j \sigma_x \otimes g_j(b_j+b^{\dagger}_j),
\end{equation}
where the system Hamiltonian takes the generic form $H_S(t)=\omega(t)\sigma_z+h(t)\sigma_x$, being concretely $\omega(t)=\omega_0$, $h(t)=\Omega\sin(\omega t)$, and we can easily identify $A=\sigma_x$, $B=\frac{1}{\sqrt{a}}\sum\limits_j g_j(b_j+b^{\dagger}_j)$ in the interaction Hamiltonian, as explained in the main text. Because of the factor $g^2=a$ appearing in Eq.\eqref{time-dep-Markov} in front of the terms due to the interaction with the environment, the factor $g=\sqrt{a}$ will be absorbed in the interaction Hamiltonian such that $\sqrt{a}B\rightarrow B$, in order to ease the notation.

The spectral decomposition of the system Hamiltonian reads
$H_S(t)=\sum\limits_{n=e,g}E_n(t)|n_t\rangle\langle n_t|$, with
\begin{equation}
\label{eigenvect}
\begin{array}{ccc}
|e_t\rangle=\left(
\begin{array}{c}
\cos(\varphi/2) \\
\sin(\varphi/2) \\
\end{array}
\right),  & &
|g_t\rangle=\left(
\begin{array}{c}
-\sin(\varphi/2) \\
\cos(\varphi/2) \\ 
\end{array}
\right),
\end{array}
\end{equation}
$E_{e,g}(t)=\pm E(t)\equiv \pm\sqrt{\omega^2(t)+h^2(t)}$ and $\tan\varphi(t)=h(t)/\omega(t)$.

Let us start by evaluating the two-time correlation function of the environment,
\begin{equation}
\label{correlation_discrete}
R(x)=tr_B[\tilde{B}(x)B\rho_{th}]=\sum\limits_{j}\frac{g_{j}^2}{1-e^{-w_j/T_B} }(e^{-\imath w_j x} + e^{\imath w_j x- w_j/T_B}) .
\end{equation}
As usual, we assume a continuous spectral density for the environment, taking the limit $\sum\limits_j g_j^2\rightarrow \int\limits_0^{\infty}dw J(w)$, hence
\begin{equation}
R(x)=\int\limits_0^{\infty}dw\frac{J(w)}{1-e^{-w/T_B} }(e^{-\imath w x} + e^{\imath w x-w/T_B}),
\end{equation}
where $J(w)=a we^{-w/w_c}$ and $a,w_c>0$. The one-sided Fourier transform of the environment correlation function in Eq.\eqref{general_one_sided} reads 
\begin{equation}
\Gamma_p(t)=\int\limits_0^{\infty}dx R(x)e^{\imath x\Omega_p(t)},
\end{equation}
with the set of Bohr frequencies $\{ \Omega_p(t)\}_p=\{0,\pm 2 E(t) \}$. Using the well-know relation $\int\limits_0^{\infty} dx e^{\imath w x}=\imath\mathcal{P}(1/w)+\pi\delta(w)$, we directly obtain for each Bohr frequency
\begin{equation}
\label{general_coeff}
\begin{aligned}
\gamma_0(t)&=4\pi a T_B , \\
\gamma_+(t)&=2\pi J(2 E(t))\bar{n}_{BE}(2E(t);T_B) ,\\
\gamma_-(t)&=2\pi J(2E(t))(1+\bar{n}_{BE}(2E(t);T_B)) , \\
S_0(t)&= -  \mathcal{P}\int\limits_0^{\infty}dw \frac{J(w)}{w} , \\
S_{\pm}(t)&= \int\limits_0^{\infty}dw J(w)\Big( (1+\bar{n}_{BE}(w))\mathcal{P}\Big( \frac{1}{-w\mp 2E(t)}\Big) + \bar{n}_{BE}(w)\mathcal{P}\Big( \frac{1}{w\mp 2E(t)}\Big) \Big), \\
\end{aligned}
\end{equation}
where $\bar{n}_{BE}(w;T_B)$ is the mean occupation number given by the Bose-Einstein distribution. 
The jump operators in the Schrödinger picture are $L_p(t)=U_S(t)|m_0\rangle\langle m_t|\sigma_x|n_t\rangle\langle n_0|U_S^{\dagger}(t)$, where $p=(n,m)$, thus we obtain three jumps

\begin{equation}
\begin{aligned}
L_0(t)&=\sum\limits_{n=e,g}U_S|n_0\rangle\langle n_t|\sigma_x|n_t\rangle\langle n_0|U_S^{\dagger}, \\
L_+(t)&=U_S|e_0\rangle\langle e_t|\sigma_x|g_t\rangle\langle g_0|U_S^{\dagger},\\ 
L_-(t)&=U_S|g_0\rangle\langle g_t|\sigma_x|e_t\rangle\langle e_0|U_S^{\dagger} , \\  
\end{aligned}
\end{equation}
where
\begin{equation}
\begin{aligned}
\langle e_t|\sigma_x|e_t\rangle=\langle g_t|\sigma_x|g_t\rangle=\sin\varphi(t), \\
\langle e_t|\sigma_x|g_t\rangle=\langle g_t|\sigma_x|e_t\rangle=\cos\varphi(t). \\
\end{aligned}
\end{equation}
For this model no analytical solutions are known for the time-evolution operator, therefore we need to rely on numerical integration of the equations of motion for computing $U_S(t)$. We start from
\begin{equation}
\label{schr_eq}
\imath\frac{d}{dt}U_S(t,0)=H_S(t)U_S(t,0)
\end{equation}
and make use of the parametrisation 
\begin{equation}
\label{time-ev}
U_S=\left(
\begin{array}{cc}
\alpha & \beta \\
-\beta^{*}e^{\imath\varphi} & \alpha^{*}e^{\imath\varphi} \\
\end{array}
\right)
\end{equation}
valid for every $2\times 2$ unitary matrix, where $\det(U_S)=e^{\imath\varphi}$ and $|\alpha|^2+|\beta|^2=1$ (see for example \cite{Sakurai}). From Eq.\eqref{schr_eq} we arrive at
\begin{equation}
\begin{aligned}
\imath\frac{d}{dt}\alpha&=\omega_0\alpha-\Omega\sin(\omega t)\beta^{*}e^{\imath\varphi} ,\\
\imath\frac{d}{dt}\beta&=\omega_0\beta+\Omega\sin(\omega t)\alpha^{*} e^{\imath\varphi} ,\\
\imath\big(\frac{d}{dt}\alpha^{*}+\imath\frac{d}{dt}\varphi \alpha^{*}\big)e^{\imath\varphi}&=\Omega\sin(\omega t)\beta - \omega_0\alpha^{*}e^{\imath\varphi} ,\\
-\imath\big(\frac{d}{dt}\beta^{*}+\imath\frac{d}{dt}\varphi \beta^{*}\big)e^{\imath\varphi}&=\Omega\sin(\omega t)\alpha + \omega_0\beta^{*}e^{\imath\varphi} .\\
\end{aligned}
\end{equation}
Combining these equations we  easily get $\frac{d}{dt}\varphi=0$, which implies $\varphi=0$ from the initial condition $\alpha(0)=1,\beta(0)=0$. Hence, the solution is determined by the set of coupled differential equations,
\begin{equation}
\begin{aligned}
\imath\frac{d}{dt}\alpha &=\omega_0\alpha-\Omega\sin(\omega t)\beta^{*} ,\\
\imath\frac{d}{dt}\beta &=\omega_0\beta+\Omega\sin(\omega t)\alpha^{*}  .\\
\end{aligned}
\end{equation}
Therefore, using the explicit form of $U_S$ we arrive at the jump operators expressed in Schrödinger picture

\begin{equation}
\begin{aligned}
L_0(t)&=\sin\varphi(t)\Big( (|\alpha|^2-|\beta|^2)\sigma_z-2\alpha\beta\sigma_+ -2\alpha^*\beta^*\sigma_-  \Big), \\
L_+(t)&=\cos\varphi(t)\Big( \alpha\beta^*\sigma_z + \alpha^2\sigma_+ -\beta^{* 2}\sigma_-  \Big), \\
L_-(t)&=\cos\varphi(t)\Big( \alpha^*\beta\sigma_z + \alpha^{* 2}\sigma_- -\beta^{ 2}\sigma_+  \Big), \\
\end{aligned}
\end{equation}
and consequently the Lamb shift Hamiltonian reads
\begin{equation}
H_{LS}(t)=-\frac{S_+-S_-}{2}\cos^2\varphi(t)\Big( (|\alpha|^2-|\beta|^2)\sigma_z-2\alpha\beta\sigma_+ -2\alpha^*\beta^*\sigma_-  \Big).
\end{equation}
Putting all together we finally arrive at the time-dependent Markovian master equation in Eq.\eqref{generic_TDME}.

\section{Derivation of the master equation for the interacting two-level systems}

In this Appendix, the master equation Eq.\eqref{me_twoqubits} is derived. The starting point is the full Hamiltonian 
\begin{equation}
H(t)=H_S(t) + \sum\limits_j w_j b^{\dagger}_j b_j +\sum\limits_j \sigma_x^{(2)} \otimes g_j(b_j+b^{\dagger}_j),
\end{equation}
where now 
\begin{equation}
H_S(t)= \frac{\omega_1(t)}{2}\sigma_z^{(1)}+\frac{\omega_2}{2}\sigma_z^{(2)}+\lambda(\sigma_+^{(1)}\sigma_-^{(2)}+h.c),
\end{equation}
with $\omega_1(t)=\omega_2+\delta\sin(\eta t)$. As explained in the main text, the condition $\omega_2>\delta$ is assumed.
The environment has the same structure as the one previously analysed in Appendix A, however here we consider the case $T_B=0$, while keeping an ohmic spectral density $J(w)=a w e^{-w/w_c}$. Let us start from the environment correlation function,  in the continuum limit we have
\begin{equation}
R(x)=tr[\tilde{B}(x)B\rho_{th}]=\int\limits_0^{\infty} dw J(w) e^{-\imath w x},
\end{equation}
therefore, utilising the relation $\Gamma(\Omega)=\int\limits_0^{\infty} dx R(x) e^{\imath\Omega x}=\frac{1}{2}\gamma(\Omega)+\imath S(\Omega)$, we come up with
\begin{equation}
\label{twoqubit_corr_app}
\begin{aligned}
\gamma(\Omega)&=\begin{cases}
2\pi J(\Omega)  \text{\space\space if $\Omega>0$}, \\
0  \text{\space\space\space\space\space\space\space\space\space\space otherwise},\\
\end{cases} \\
S(\Omega)&= \mathcal{P}\int\limits_0^{\infty} dw \frac{J(w)}{\Omega-w},
\end{aligned}
\end{equation}
where $\Omega$ will correspond to the instantaneous Bohr frequencies. Notice that,  at zero temperature the environment correlation function can be easily computed and it is given by $R(t)= \frac{ a w_c^2}{(1-\imath t w_c)^2} $.

In order to compute the jump operators, the eigendecomposition of $H_S(t)$ is needed. In particular, the instantaneous eigenvalues are
\begin{equation}
\begin{array}{ccc}
E_0(t)=-\omega_+(t), & & E_2(t)=\sqrt{\lambda^2+\omega_-^2(t)}, \\
E_1(t)=-\sqrt{\lambda^2+\omega_-^2(t)}, & & E_3(t)=\omega_+(t), \\
\end{array}
\end{equation}
where $\omega_{\pm}(t)=\frac{1}{2}(\omega_1(t)\pm\omega_2)$, whereas
\begin{equation}
\begin{aligned}
|E_0(t)\rangle &=|0,0\rangle, \\ 
|E_1(t)\rangle &=-\sin\theta(t)|1,0\rangle + \cos\theta(t)|0,1\rangle, \\ 
|E_2(t)\rangle &=\;\;\;\cos\theta(t)|1,0\rangle + \sin\theta(t)|0,1\rangle, \\ 
|E_3(t)\rangle &=|1,1\rangle , \\ 
\end{aligned}
\end{equation}
are the eigenvectors, where $\tan(2\theta(t))=\lambda/\omega_-(t)$
and  we made use of the eigenbasis $\{|1\rangle,|0\rangle \}$ of $\sigma_z$ in $\mathbb{C}^2$. The ordering $E_{j+1}>E_j$ is valid if we assume the condition $\omega_1(t)\omega_2 >\lambda^2, \forall t$, or equivalently $(\omega_2-\delta)\omega_2>\lambda^2$, so that no crossover among energy levels occurs. In general, for positive coupling we can individuate three different intervals
\begin{enumerate}
\item $0<\lambda< \sqrt{(\omega_2-\delta)\omega_2}$. In this case, the condition $\lambda^2<\omega_1(t)\omega_2$ is satisfied for all $t$, the spectrum is smooth and there are no crossovers.
\item $\sqrt{(\omega_2-\delta)\omega_2}\le\lambda\le \sqrt{(\omega_2+\delta)\omega_2} $. There are multiple times $t^*$ such that $\lambda^2=\omega_1(t^*)\omega_2$ and level crossings take place.
\item $\lambda>\sqrt{(\omega_2+\delta)\omega_2} $. The spectrum is smooth as the condition $\lambda^2>\omega_1(t)\omega_2$ is fulfilled.
\end{enumerate}
The jump operators in the Schrödinger pictures are in the form
\begin{equation}
\{ L_{nm}(t)=U_S(t)|E_m(0)\rangle\langle E_m(t)|\sigma_x^{(2)}|E_n(t)\rangle\langle E_n(0)|U_S^{\dagger}(t)\}
\end{equation} 
and correspond to the instantaneous Bohr frequencies $\{E_n(t)-E_m(t)\}$. We start by computing the non-zero matrix elements $\langle E_m(t)|\sigma_x^{(2)}|E_n(t)\rangle$, which enter the master equation,
\begin{equation}
\begin{array}{ccc}
\langle E_1(t)|\sigma_x^{(2)}|E_0(t)\rangle=\cos\theta(t), & &\langle E_3(t)|\sigma_x^{(2)}|E_1(t)\rangle=-\sin\theta(t),\\
\langle E_2(t)|\sigma_x^{(2)}|E_0(t)\rangle=\sin\theta(t), & &\langle E_3(t)|\sigma_x^{(2)}|E_2(t)\rangle= \cos\theta(t) .\\
\end{array}
\end{equation}
The remaining vectors $|\psi_n(t)\rangle\equiv U_S(t)|E_n(0)\rangle$ for any $n$ can be directly obtained solving the Schrödinger equation for $|\psi(t)\rangle=(\psi_{11},\psi_{10},\psi_{01},\psi_{00})^T \in \mathbb{C}^4$ and imposing the initial condition $|\psi(0)\rangle= |E_n(0)\rangle$. From $\imath\frac{d}{dt}|\psi(t)\rangle=H_S(t)|\psi(t)\rangle$ we obtain the following set of equations
\begin{equation}
\label{twoqub_param}
\begin{aligned}
\imath\frac{d}{dt}\psi_{11}&=\omega_+(t)\psi_{11}, \\
\imath\frac{d}{dt}\psi_{00}&=-\omega_+(t)\psi_{00}, \\
\imath\frac{d}{dt}\psi_{10}&=\omega_-(t)\psi_{10} + \lambda\psi_{01}, \\
\imath\frac{d}{dt}\psi_{01}&=-\omega_-(t)\psi_{01} + \lambda\psi_{10} \\
\end{aligned}
\end{equation}
and indicate with $|\psi_n(t)\rangle$ the solution of Eq.\eqref{twoqub_param} with initial condition given by $|E_n(0)\rangle$. Clearly, the last equation needs to be solved numerically.
Therefore, we obtain the jump operators $\{L_{a}(t), L_{b}(t), L_{a}^{\dagger}(t), L_{b}^{\dagger}(t)\}$, where
\begin{equation}
\begin{aligned}
L_{a}(t)&= L_{20}+L_{31}=\sin\theta(t)\Big(|\psi_0(t)\rangle \langle\psi_2(t)|  -|\psi_1(t)\rangle \langle\psi_3(t)|  \Big), \\
L_{b}(t)&= L_{32}+L_{10}=\cos\theta(t)\Big(|\psi_0(t)\rangle \langle\psi_1(t)|  +|\psi_2(t)\rangle \langle\psi_3(t)|  \Big), \\
\end{aligned}
\end{equation}
correspondent to the Bohr frequencies $\{\Omega_a,\Omega_b,-\Omega_a,-\Omega_b \}$, with
\begin{equation}
\begin{aligned}
\label{set_bohr_twoqubits}
\Omega_{a}\equiv\Omega_{31}=\Omega_{20}=\omega_+(t) + \sqrt{\lambda^2+ \omega_-^2(t)},  \\
\Omega_{b}\equiv \Omega_{10}=\Omega_{32}=\omega_+(t) - \sqrt{\lambda^2+ \omega_-^2(t)}. \\
\end{aligned}
\end{equation}
The Lamb shift contribution reads
\begin{equation}
H_{LS}(t)=\sum\limits_{p=a,b}\Big( S(\Omega_p(t))L_{p}^{\dagger}(t)L_{p}(t)+ S(-\Omega_p(t))L_{p}(t)L_{p}^{\dagger}(t)\Big).
\end{equation}
Finally, we can write down the master equation in the Schrödinger picture
\begin{equation}
\begin{aligned}
\frac{d}{dt}\rho_S=&-\imath[H_S(t)+H_{LS}(t),\rho_S]+\gamma(\Omega_{a}(t))\Big( L_{a}(t)\rho_S L_{a}^{\dagger}(t)-\frac{1}{2}\{L_{a}^{\dagger}(t)L_{a}(t),\rho_S \} \Big) \\
&+ \gamma(\Omega_{b}(t))\Big( L_{b}(t)\rho_S L_{b}^{\dagger}(t)-\frac{1}{2}\{L_{b}^{\dagger}(t)L_{b}(t),\rho_S \} \Big).
\end{aligned}
\end{equation}
Beside the already discussed conditions on the environment correlation function and weak coupling, the validity of the master equation is also determined by the following:

\begin{enumerate}
\item \textit{Full-secular approximation}. 
The explicit expression of this condition depends on the values of the parameters, in particular the coupling $\lambda$. We need to distinguish three different scenarios. According to the prescription given in Eq.\eqref{general_secular}, we consider the differences between instantaneous Bohr frequencies in Eq.\eqref{set_bohr_twoqubits}, i.e. the set
\begin{equation}
\label{eq_bohr_diff}
\mathcal{F}\equiv \{\pm 2\sqrt{\lambda^2+\omega_-^2(t)},\pm  2\omega_+(t), \pm 2\Big(\omega_+(t)+\sqrt{\lambda^2+\omega_-^2(t)}\Big),\pm 2\Big(\omega_+(t)-\sqrt{\lambda^2+\omega_-^2(t)}\Big) \}.
\end{equation}
\begin{enumerate}
\item  $0<\lambda<\sqrt{(\omega_2-\delta)\omega_2}$. In this case, there are no crossovers among energy levels, in particular the set in Eq.\eqref{eq_bohr_diff} shows no critical points, i.e. $\frac{d}{dt}\varphi_{p q}(t)\neq 0, \forall t$. Hence we can use the first condition, i.e. 
\begin{equation}
\min\limits_{p,q;t} |\Omega_{p}(t)-\Omega_{q}(t)|\gg a C,
\end{equation}
where $C=\int\limits_0^{\infty} dt |R(t)|=\frac{\pi}{2} w_c$ can be computed analytically in this case (notice that here we have extracted the factor $\sqrt{a}\equiv g$ from the interaction term).
The resulting condition, after the minimisation over the set and time $t$ is
\begin{equation}
\min\{2\lambda, 2\omega_2-\delta-2\sqrt{\lambda^2+(\delta/2)^2} \} \gg a C\, .
\end{equation}

\item $\sqrt{(\omega_2-\delta)\omega_2 }\le \lambda\le \sqrt{(\omega_2+\delta)\omega_2} $. In this case, multiple crossovers take place at different times when $\omega_+(t)=\sqrt{\lambda^2+\omega_-^2(t)}$. We split the set $\mathcal{F}$ between phases with and without critical points, say  $\mathcal{F}=\mathcal{F}_c \cup \mathcal{F}_{nc}$, where
\begin{equation}
\begin{aligned}
\mathcal{F}_{nc}&= \{\pm 2\sqrt{\lambda^2+\omega_-^2(t)},\pm  2\omega_+(t), \pm 2\Big(\omega_+(t)+\sqrt{\lambda^2+\omega_-^2(t)}\Big)\}, \\
\mathcal{F}_{c}&= \{\pm 2\Big(\omega_+(t)-\sqrt{\lambda^2+\omega_-^2(t)}\Big)  \}. 
\end{aligned}
\end{equation}
For the first set, the usual condition is valid, i.e. $\min\limits_{p,q;t} |\Omega_{p}(t)-\Omega_{q}(t)|\gg a C$, which translates into
\begin{equation}
\min\{2\omega_2-\delta,2\lambda \}\gg aC,
\end{equation}
whereas for the second set we need to fulfil $\min\limits_{p,q;t_0} |\frac{d\Omega_p}{dt}(t_0)-\frac{d\Omega_q}{dt}(t_0)|\gg a C^2$, where $t_0$ indicates the set of critical points, hence we arrive at
\begin{equation}
\eta\delta\gg 2a C^2.
\end{equation}

\item  $\sqrt{(\omega_2+\delta)\omega_2}\le \lambda$. Here the spectrum is again smooth and the minimisation over the full set of Bohr frequencies leads to
\begin{equation}
\min\{2\omega_2-\delta, 2\sqrt{\lambda^2+(\delta/2)^2}  -2\omega_2-\delta\} \gg a C\, .
\end{equation}

\end{enumerate}

\item \textit{Condition on the driving}. Following the prescription in Eq.\eqref{driving_cond}, as before we need to distinguish the three ranges for $\lambda$.

\begin{enumerate}
\item $0<\lambda<\sqrt{(\omega_2-\delta)\omega_2}$. Because of the smoothness of the spectrum, no critical points appear and the first condition in Eq.\eqref{driving_cond} can be simply used.
One can easily find 
\begin{equation}
\begin{aligned}
\max\limits_{n\neq m; t}|\langle E_n(t)|\frac{d}{dt}|E_m(t)\rangle|&=\max\limits_t \Big|\frac{d}{dt}\theta(t)\Big|\le \frac{\delta\eta}{4\lambda}, \\
\min\limits_{n\neq m; t}| E_n(t)-E_m(t)|&=\min\{2\lambda, \omega_2-\frac{\delta}{2}-\sqrt{\lambda^2+\frac{\delta^2}{4}} \}.
\end{aligned}
\end{equation}
Hence,  sufficient condition can be stated as
\begin{equation}
\frac{\delta\eta}{4\lambda}\le a^{-1} \min\{2\lambda, \omega_2-\frac{\delta}{2}-\sqrt{\lambda^2+\frac{\delta^2}{4}} \}.
\end{equation}

\item $\sqrt{(\omega_2-\delta)\omega_2 }\le \lambda\le \sqrt{(\omega_2+\delta)\omega_2} $. Both conditions in Eq.\eqref{driving_cond} have to be addressed. To this end, we split the set of $\Delta_{n m}$ such as $\mathcal{F}=\mathcal{F}_{nc}+\mathcal{F}_c$, where 
\begin{equation}
\begin{aligned}
\mathcal{F}_{nc}&=\{\pm \Delta_{2 0}, \pm \Delta_{2 1}, \pm \Delta_{3 0}, \pm\Delta_{3 1} \}, \\
\mathcal{F}_c&=\{ \pm\Delta_{3 2},\pm\Delta_{1 0} \} .\\
\end{aligned}
\end{equation}
For the second set, we see that $\alpha_{n m}(t)=0$, hence the inequality is trivially satisfied, while for the first set we obtain
\begin{equation}
\begin{aligned}
\max\limits_{n\neq m; t}|\langle E_n(t)|\frac{d}{dt}|E_m(t)\rangle|&=\max\limits_t \Big|\frac{d}{dt}\theta(t)\Big|\le \frac{\delta\eta}{4\lambda}, \\
\min\limits_{n\neq m; t}|E_n(t)-E_m(t)|&=\min\{2\lambda, 2\omega_2-\delta \},
\end{aligned}
\end{equation}
hence the condition reads 
\begin{equation}
\frac{\delta\eta}{4\lambda}\le a^{-1} \min\{2\omega_2-\delta,2\lambda \},
\end{equation}

\item $\sqrt{(\omega_2+\delta)\omega_2}\le \lambda$. The spectrum is smooth again and 
\begin{equation}
\begin{aligned}
\max\limits_{n\neq m; t}|\langle E_n(t)|\frac{d}{dt}|E_m(t)\rangle|&=\max\limits_t \Big|\frac{d}{dt}\theta(t)\Big|\le \frac{\delta\eta}{4\lambda}, \\
\min\limits_{n\neq m;t}| E_n(t)-E_m(t)|&=\min\{2\omega_2-\delta,\sqrt{\lambda^2+\frac{\delta^2}{4}}-\omega_2-\frac{\delta}{2} \}.
\end{aligned}
\end{equation}
Hence, the sufficient condition can be stated as
\begin{equation}
\frac{\delta\eta}{4\lambda}\le a^{-1} \min\{2\omega_2-\delta,\sqrt{\lambda^2+\frac{\delta^2}{4}}-\omega_2-\frac{\delta}{2} \}.
\end{equation}

\end{enumerate}

\end{enumerate}

\paragraph{Derivation of the local master equation in the case of weak coupling between qubits.}

We show here how to derive a local master equation for the interacting qubits. The starting point is the derivation outlined in Theorem \ref{th_main}.
In our case, we simply consider a further condition when studying the limit $g\rightarrow 0$ of the Nakajima-Zwanzig kernel in Eq.\eqref{kg_0}, namely that $\lambda=O(g)$. This allows us to perform in the first place the expansion $U_S=U_S^{(1)}\otimes U_S^{(2)}+ O(g)$, where $U_S^{(1)}(t)=e^{-\imath\int\limits_0^t ds\omega_1(s)/2 \sigma_z^{(1)}}$, $U_S^{(2)}(t)=e^{-\imath\frac{\omega_2}{2} t \sigma_z^{(2)}}$, obtaining
\begin{equation}
\begin{aligned}
\tilde{A}(t)&=(U_S^{(1)}\otimes U_S^{(2)})^{\dagger}\sigma_x^{(2)}(U_S^{(1)}\otimes U_S^{(2)})+ O(g) \\
&=e^{\imath \frac{\omega_2}{2} t \sigma_z^{(2)}}\sigma_x^{(2)} e^{-\imath \frac{\omega_2}{2} t \sigma_z^{(2)}}+ O(g)\\
&= e^{-\imath\omega_2 t}\sigma_- +e^{+\imath\omega_2 t}\sigma_+ + O(g) .
\end{aligned}
\end{equation}
At this point, the derivation can be easily carried on by considering the time rescaling and the limit $g\rightarrow 0$, as already illustrated. The resulting master equation reads
\begin{equation}
\frac{d}{dt}\rho_S=-\imath[H_S(t)+H_{LS},\rho_S]+\gamma(\omega_2)\Big(\sigma_-^{(2)}\rho_S\sigma_+^{(2)}-\frac{1}{2}\{ \sigma_+^{(2)}\sigma_-^{(2)},\rho_S\} \Big),
\end{equation}
where $H_{LS}=\frac{1}{2}(S(\omega_2)-S(-\omega_2))\sigma_z^{(2)}$.
Underlying this equation, a full-secular approximation is also needed and reads $2\omega_2\gg a C$. On the other hand, there is no constraint on $\omega_1(t)$.

\section{Environmental correlation functions}

In this Appendix, we study the boundedness condition for the correlation function
\begin{equation}
\label{correl_inf}
R(t)=\int\limits_0^{\infty}dwJ(w)\Big[e^{-\imath w t}(1+\bar{n}_{BE}(w)) + e^{\imath w t} \bar{n}_{BE}(w)\Big],
\end{equation}
with $J(w)=a w e^{-w/w_c}$, $\bar{n}_{BE}(w)=(e^{\beta w}-1)^{-1}$ and $\beta\neq 0$.
\begin{lemma}
The function $R:\mathbb{R}^+\rightarrow \mathbb{C}$ in Eq.\eqref{correl_inf} is integrable in $\mathbb{R}^+$.
\end{lemma}
\begin{proof}
First, we notice that $R(t)$ can be written in terms of the trigamma function $\psi^{(1)}(z)=\int\limits_0^{\infty} dx x e^{-zx}/(1-e^{-x}) $ for $z\in\mathbb{C}$, $\Re (z)>0$ (see \cite{Handbook}) as
\begin{equation}
R(t)=\frac{a}{\beta^2}\Big(\psi^{(1)}\Big(\frac{1}{\beta w_c}+\imath\frac{t}{\beta}\Big) + \psi^{(1)}\Big(\frac{1}{\beta w_c}-\imath\frac{t}{\beta}+1\Big)\Big).
\end{equation}
Because $\psi^{(1)}(z)$ is analytic in any $K=\{ z=(x,y) \in\mathbb{C}| x>0, y\in I \}$, $R(t)$ is locally integrable in any compact interval $I\subset \mathbb{R}^+$.

Now we need to study the asymptotic behaviour. In order to do so,
it is sufficient to show that $|R(t)|=O(t^{-2})$ for $t\rightarrow \infty$.
To this end, using the recurrence formula $\psi^{(1)}(z+1)=\psi^{(1)}(z)-\frac{1}{z^2}$ we obtain
\begin{equation}
R(t)=\frac{a}{\beta^2}\Big(\psi^{(1)}\Big(\frac{1}{\beta w_c}+\imath\frac{t}{\beta}\Big) + \psi^{(1)}\Big(\frac{1}{\beta w_c}-\imath\frac{t}{\beta}\Big)-\Big(\frac{1}{\beta w_c}-\imath\frac{t}{\beta} \Big)^{-2}\Big).
\end{equation}
The last contribution is clearly integrable as $\Big|\frac{1}{\beta w_c}-\imath\frac{t}{\beta} \Big|^{-2}=O(t^{-2})$ for $t\rightarrow \infty$, whereas for the first two one can easily see that also
\begin{equation}
\psi^{(1)}\Big(\frac{1}{\beta w_c}+\imath\frac{t}{\beta}\Big) + \psi^{(1)}\Big(\frac{1}{\beta w_c}-\imath\frac{t}{\beta}\Big)=O(t^{-2}),
\end{equation}
using the asymptotic expansion $\psi^{(1)}(z)\approx\frac{1}{z}+\frac{1}{2z^2}+... $ for $z\rightarrow\infty$. 
\end{proof}

\section{Tensor network simulations}

\begin{figure}[t]
\centering
\includegraphics[width=10cm]{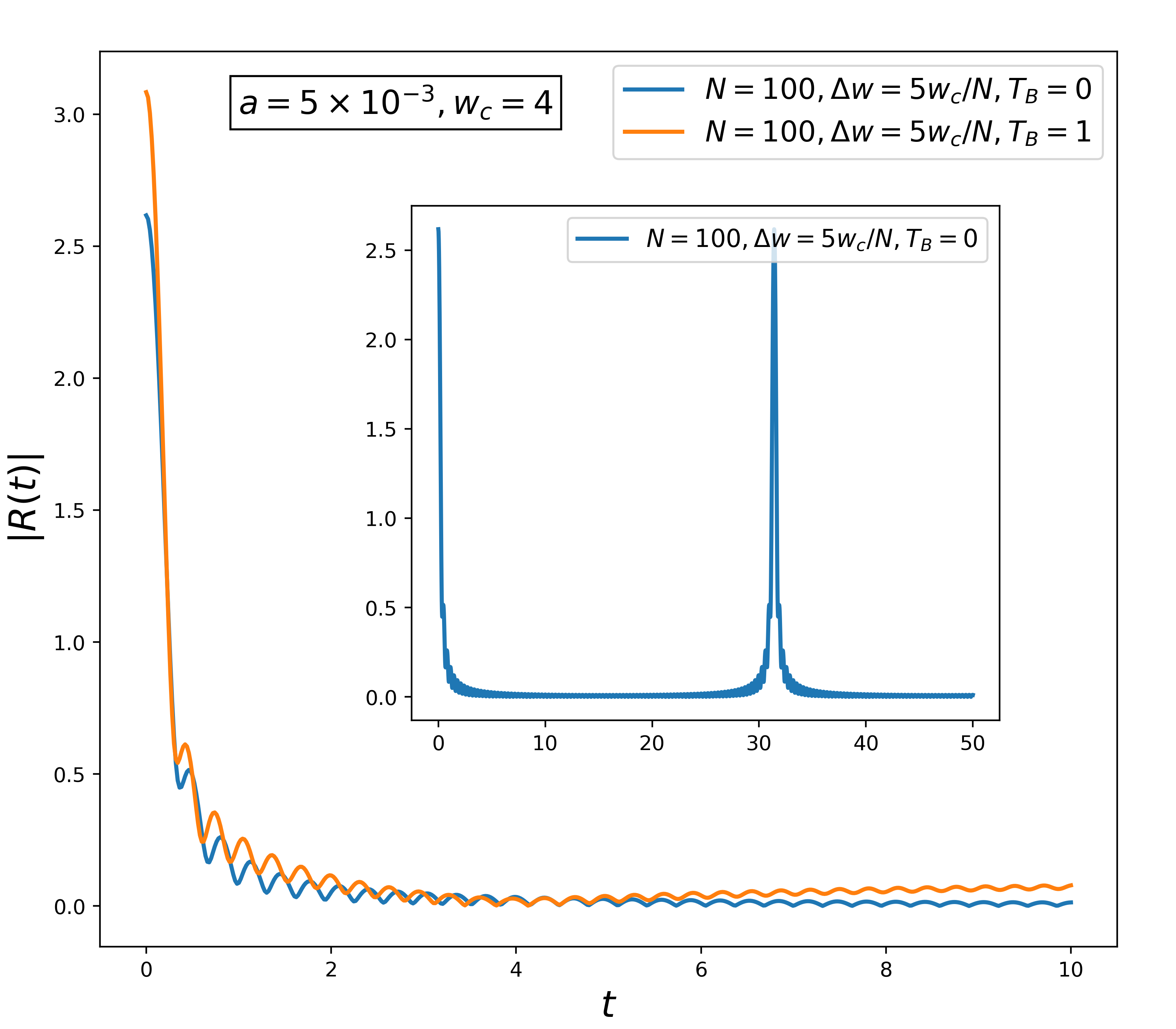}
\caption{Absolute value of the two-time correlation function $R(t)=tr_B[\tilde{B}(t)B\rho_{th}]$ from Eq.\eqref{correlation_discrete}. The spectral density parameters $a=5\times 10^{-3}, w_c=4$ are fixed, $N$ is the number of oscillators, $T_B$ their temperature, $\Delta w=w_{max}/N$ with $w_{max}$ the highest frequency of the environment modes. The main panel shows the rapid decay of the correlation function for a sufficiently small discretisation of the spectrum, while in the internal panel we can see the periodicity due to the finiteness of the environment, which   manifests for larger times.}
\label{fig:Correlation}
\end{figure}

In order to analyse the range of validity of our result, 
the master Eq.\eqref{time-dep-Markov} is benchmarked against numerically exact results given by tensor network simulations (\cite{TN_1, TN_2, TN_3}). In this section, we introduce the method employed for the numerical simulations. We model the total system  as a system S coupling to a set of $N$ harmonic oscillators in a star configuration, with total Hamiltonian in natural units ($\hbar=1$)
\begin{equation}
\label{general_spinbos}
H(t)=H_S(t)+\sum\limits_j^N w_j b_j^{\dagger}b_j + \sum\limits_j^N \sigma_x\otimes B_j,
\end{equation} 
with $B_j=g_j(b_j + h.c.)$, $H_S(t)$ generic  and
the bosonic creation/annihilation operators satisfy the usual canonical commutation relations $[b_i,b^{\dagger}_j]=\delta_{ij}$, $[b_i,b_j]=0$ , while $w_j>0$ without loss of generality. We are interested in studying the evolution of the density matrix $\rho(t)$, according to the Von Neumann equation $\frac{d}{dt}\rho(t)=-\imath[H(t),\rho(t)]$.
First, we consider some features concerning the environment.
\begin{enumerate}
\item \textit{Finiteness of the environment}. In order to achieve a full Markovian dynamics the environment must have
infinite degrees of freedom. In practice, given that we are considering a collection of $N$ harmonic oscillators, we can avoid the recurrence due to the finiteness of the environment bounding the time evolution with $t < t_{max}$, where $t_{max}$ has to be chosen in such a way that no periodicity of the correlation function $R(t)=tr_B[\tilde{B}(t)B\rho_{th}]$ is observed in the considered time scale. Fig.\eqref{fig:Correlation} shows an example of this behaviour.  For our simulations we consider $N=100$ fixed.
\item \textit{Choice for the coupling constants}.
The coefficients $g_j$ are determined by their relation with the spectral density. In the continuum limit,  $\sum\limits_j g_j^2 \equiv \int\limits_0^{\infty} dw J(w)$ with $J(w)=awe^{-w/w_c}$ ohmic spectral density used throughout this work. Let us introduce an upper limit $w_{max}$ for the integral, because of the finite number of the environment degrees of freedom, and replace it with a discrete sum as
\begin{equation}
\int\limits_0^{\infty} dw J(w)\approx\sum\limits_j a w_j e^{-w_j/w_c}\Delta w ,
\end{equation}
where $w_j=j\Delta w$, $\Delta w$ is the integration step and $j=1,...,N$ with $w_{max}=\Delta w N$, hence we obtain $g_j^2\approx a w_j e^{-w_j/w_c}\Delta w $. 
This approximation is affected by two sources of error. First, we have the error due to the truncation at $w_{max}$ 
\begin{equation}
\epsilon_1=\Big|\int\limits_0^{\infty}dw J(w)-\int\limits_0^{w_{max}}dw J(w)\Big|= a w_c^2\Big(\frac{w_{max}}{w_c}+1 \Big) e^{-w_{max}/w_c}.
\end{equation}
As long as $w_{max}\gtrsim w_c$, this error can be easily taken under control for sufficiently small $a$. The second one is determined by the discretisation of the integral, in particular using the general result for approximating integrals with the Riemann right sum
\begin{equation}
\epsilon_2=\Big|\int\limits_0^{w_{max}}dw J(w)-\sum\limits_{j=1}^N \Delta w J(w_j)\Big|\le \frac{w_{max}^2}{2N}\max\limits_{w\in[0,w_{\max}]}|\frac{d}{dw} J(w)|\equiv \frac{a w_{max}^2}{2N}, 
\end{equation}
for $w_{max}\gtrsim w_c$. If $\Delta w=w_{max}/N$ is too large, then from the physical point of view the system does not 'see' the environment as being made up of a continuum of modes and the dynamics ceases to be Markovian since the beginning of the evolution.

The integration step $\Delta w$ is determined from $w_{max}$, once $N$ and $w_c$ have been chosen, while the weak coupling condition is guaranteed by $a\sim 10^{-3}$ for example, as already explained in the previous sections.

\item \textit{Truncation of the environment Hilbert space}. The local Hilbert space dimension $d_j$ of the environment is chosen by looking at the thermal mean occupation value $\langle b_j^{\dagger}b_j \rangle_{th}$ as reference, therefore the low frequency oscillators (i.e. the ones with $j$ close to zero) have generally higher dimension, while the high frequency ones can be efficiently modelled with $d_j=2$ or $3$. 

\end{enumerate}

\begin{figure}[htb]
\centering
\includegraphics[width=10cm]{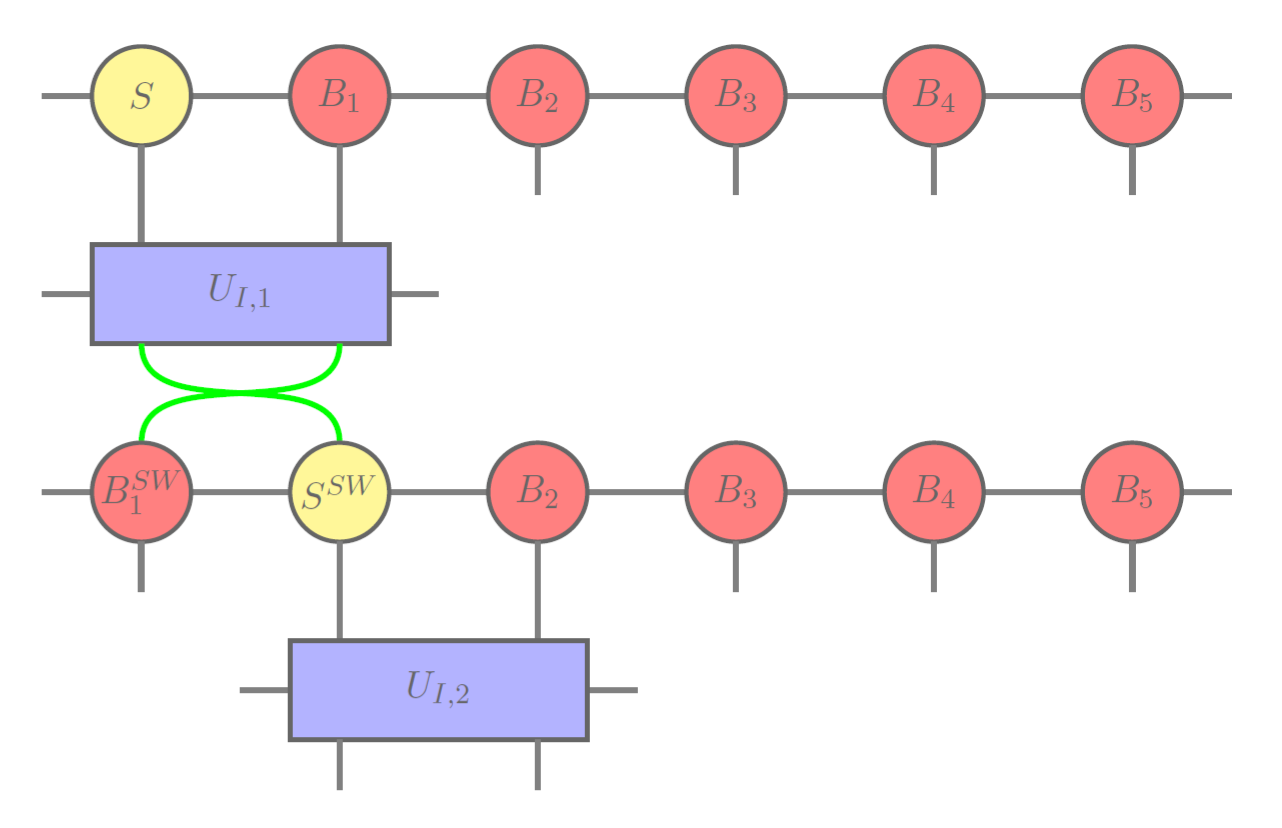}
\caption{Scheme of the application of unitary operators $U_{I,j}=e^{-\imath H_{I,j}\frac{\delta t}{2}} $ as described in the main text. The MPS is represented by the chain above, S in the initial site is the system while the remaining are $N=5$ harmonic oscillators. After the application of the first unitary gate  $U_{I,1}$ the physical legs are swapped (green curves). This way, after each step the system is moved to the right along the chain and the interaction is kept local. When the system arrives at the end of the chain, we repeat the procedure moving to the left and applying the inverse swaps, until we return to the initial configuration.}
\label{fig: TN_scheme}
\end{figure}

We study the dynamics by means of Time-Evolving Block Decimation on the purified density matrix, as described by the following procedure:
\begin{enumerate}
\item
The total time-evolution operator between $t$ and $t+\delta t$ is written in the typical form of a second-order TEBD (see \cite{TN_4,TN_5}):
\begin{equation}
U(t+\delta t,t)\approx \exp\Big\{-\imath\int\limits_{t+\delta t/2}^{t+\delta t} ds H_S(s)\Big\} e^{-\imath H_B \frac{\delta t}{2} }e^{-\imath H_I\delta t} e^{-\imath H_B\frac{\delta t}{2}}\exp\Big\{-\imath\int\limits_t^{t+\delta t/2} ds H_S(s)\Big\}
\end{equation}
where we can naturally factorise the terms made by internally commuting contributions as $e^{-\imath\frac{\delta t}{2} H_B}=\Bigg[ \prod\limits_{j=1}^N e^{-\imath\frac{\delta t}{2} H_{B j}}\Bigg]$. Note that we have also removed the time-ordering from $\mathcal{T}\exp\Big\{-\imath\int\limits_t^{t+\delta t/2} ds H_S(s)\Big\}$: as shown in \cite{TN_5}, this produces a further error bounded by  a contribution of order $\delta t ^2||H_S(t)||^2$. For our simulations we have used $\delta t\approx 10^{-2}$.

\item The interacting contribution is further split into $e^{-\imath H_I \delta t}=e^{-\imath  H_{I 1} \frac{\delta t}{2}}...e^{-\imath  H_{I N}\frac{\delta t}{2}}e^{-\imath H_{I N}\frac{\delta t}{2}}...e^{-\imath H_{I 1}\frac{\delta t}{2}}$ without producing additional errors.
This decomposition allows us to efficiently treat long-range interactions by means of swap gates after each contraction with the unitary operator $U_{I,j}=e^{-\imath H_{I j}\frac{\delta t}{2}}$, as described in Fig.\eqref{fig: TN_scheme} and more in detail in \cite{TN_7}.

\item As anticipated, we can consider a local purification $|\psi\rangle$ of the density matrix doubling each local Hilbert space by $\mathcal{H}_n\rightarrow \mathcal{H}_n\otimes \mathcal{H}_n^{aux}$ with an auxiliary Hilbert space isomorphic to the physical one. In this manner, we can safely keep under control the positivity of the density matrix after each compression.
Hence, the MPS representation of our state takes the generic form
\begin{equation}
|\psi\rangle=\sum\limits_{\sigma,\sigma^{'}}^{d_s}\sum\limits_{n_1,n_1^{'}}^{d_1}... \sum\limits_{n_N,n_N^{'}}^{d_N}S^{\sigma \sigma^{'}} B_1^{n_1 n_1^{'}}...B_N^{n_N n_N^{'}}\big[|\sigma\rangle\otimes|\sigma^{'}\rangle\big]\otimes ... \otimes\big[|n_N\rangle\otimes|n_N^{'}\rangle\big],
\end{equation}
where $d_s=2$ for a single qubit, $d_j$ the local $j-th$ bath Hilbert space dimension and naturally $\rho=tr_{aux}[|\psi\rangle\langle\psi|]$.
Let us show that by simply considering $h_n\rightarrow h_n\otimes 1_{d_n}^{aux}$, where $h_n$ is a generic spin or bath operator acting only on the local  Hilbert space $\mathcal{H}_n$, the dynamics generated is exactly the same.
For the free contributions to the Hamiltonian we have
\begin{equation}
\begin{cases}
U_S\rightarrow \mathcal{U}_S=\exp\Big\{-\imath\int\limits_t^{t+\delta t} ds H_S(s)\Big\}\otimes 1_{d_s}^{aux} , \\
U_{B j}\rightarrow \mathcal{U}_{B j}=e^{-\imath\frac{\delta t}{2} H_{B j}}\otimes 1_{d_j}^{aux} , \\
\end{cases}
\end{equation}
then clearly $|\psi\rangle\rightarrow |\psi^{'}\rangle=\mathcal{U}|\psi\rangle$  corresponds to $\rho\rightarrow \rho^{'}=U\rho U^{\dagger}$.
The same consideration can be applied to the interacting contributions. First, notice that for $H_{I j}=\sigma_x\otimes B_j$ is $U_{I j}=1_{d_s}\otimes \cos(B_j\delta t)-\imath\sigma_x\otimes \sin(B_j\delta t)$, hence if  $H_{I j} \rightarrow  \sigma_x\otimes 1_{d_s}^{aux}\otimes B_j\otimes 1_{d_j}^{aux}$ then $U_S\rightarrow \mathcal{U}_{I j}=1_{d_s}\otimes 1_{d_s}^{aux}\otimes \cos(B_j\delta t)\otimes 1_{d_j}^{aux}-\imath\sigma_x\otimes 1_{d_s}^{aux}\otimes \sin(B_j\delta t)\otimes 1_{d_j}^{aux}$, from which we easily prove
\begin{equation}
\begin{aligned}
\rho^{'}&=tr_{aux}\Big[\mathcal{U}_{I j}|\psi\rangle\langle\psi|\mathcal{U}_{I j}^{\dagger}\Big],\\
&=\sum\limits_{n,k}\langle \phi_n^{(S)}|\langle \phi_k^{(B_j)}|\mathcal{U}_{I j}|\psi\rangle\langle\psi|\mathcal{U}_{I j}^{\dagger}|\phi_k^{(B_j)}\rangle |\phi_n^{(S)}\rangle ,\\
&=U_{I j}\sum\limits_{n,k}\langle \phi_n^{(S)}|\langle \phi_k^{(B_j)}|\psi\rangle\langle\psi|\phi_k^{(B_j)}\rangle |\phi_n^{(S)}\rangle U_{I j}^{\dagger},\\
&=U_{I j}tr_{aux}[|\psi\rangle\langle\psi|]U_{I j}^{\dagger} ,
\end{aligned}
\end{equation}
having introduced two orthonormal bases $\{|\phi_n^{(S)}\rangle \}_n,\{|\phi_n^{(B_j)}\rangle \}_n$ for, respectively, the auxiliary Hilbert spaces of $\mathcal{H}_S,\mathcal{H}_j$ and using the explicit form of $\mathcal{U}_{I j}$.
\item The initial condition is simply given by the vectorisation of the square root of the initial density matrix. For example, in the case of initial preparation of a single qubit in the ground state,
\begin{equation}
|\psi(0)\rangle=[|0\rangle\otimes |0\rangle]\otimes \Big[ \sum\limits_k \sqrt{p_k^{(1)}}|\phi_k^{(1)}\rangle\otimes|\phi_k^{(1)}\rangle\Big]\otimes...\otimes\Big[ \sum\limits_k \sqrt{p_k^{(N)}}|\phi_k^{(N)}\rangle\otimes|\phi_k^{(N)}\rangle\Big],
\end{equation}
with the thermal state for the $j$-th bath in the form $\rho_{th}^{(j)}=\sum\limits_k p_k^{(j)}|\phi_k^{(j)}\rangle\langle\phi_k^{(j)}|$.
\end{enumerate}

\bibliographystyle{quantum}
\bibliography{mybiblio}

\end{document}